\documentclass[11pt]{article}
\usepackage[margin=1in]{geometry}

\usepackage{amsmath,amsfonts,amssymb}
\usepackage{amsthm}
\usepackage{algorithm}
\usepackage{algorithmic}
\usepackage{bm}
\usepackage{graphicx}
\usepackage{multirow}
\usepackage{mathtools}
\usepackage{subcaption}
\usepackage{textcomp}
\usepackage{xcolor}
\usepackage[inline]{enumitem}
\newtheorem{theorem}{Theorem}
\DeclareMathOperator{\diag}{diag}
\DeclareMathOperator*{\argmin}{arg\,min}
\usepackage[hyphens]{url}
\usepackage{cleveref}

\title{Scalable Dual Coordinate Descent for Kernel Methods}

\begin{document}

\author{Zishan Shao\\shaoz20@wfu.edu\\Wake Forest University\\\\ Aditya Devarakonda\\devaraa@wfu.edu\\Wake Forest University}
\date{}
\maketitle

\begin{abstract}
    Dual Coordinate Descent (DCD) and Block Dual Coordinate Descent (BDCD) are important iterative methods for solving convex optimization problems.
    In this work, we develop scalable DCD and BDCD methods for the kernel support vector machines (K-SVM) and kernel ridge regression (K-RR) problems.
    On distributed-memory parallel machines the scalability of these methods is limited by the need to communicate every iteration.
    On modern hardware where communication is orders of magnitude more expensive, the running time of the DCD and BDCD methods is dominated by communication cost.
    We address this communication bottleneck by deriving $s$-step variants of DCD and BDCD for solving the K-SVM and K-RR problems, respectively.
    The $s$-step variants reduce the frequency of communication by a tunable factor of $s$ at the expense of additional bandwidth and computation.
    The $s$-step variants compute the same solution as the existing methods in exact arithmetic.
    We perform numerical experiments to illustrate that the $s$-step variants are also numerically stable in finite-arithmetic, even for large values of $s$.
    We perform theoretical analysis to bound the computation and communication costs of the newly designed variants, up to leading order.
    Finally, we develop high performance implementations written in C and MPI and present scaling experiments performed on a Cray EX cluster.
    The new $s$-step variants achieved strong scaling speedups of up to $9.8\times$ over existing methods using up to $512$ cores.
\end{abstract}

\section{Introduction}
Optimization methods, particularly Dual Coordinate Descent (DCD) and Block Dual Coordinate Descent (BDCD), are fundamental to efficiently training nonlinear machine learning model on large-scale datasets.
These iterative methods are particularly useful in solving regression and classification (both regularized and unregularized) optimization problems that arise in various research areas such as biology, computer vision, biophysics and healthcare.
Given the volume of data generated in these scientific disciplines and the availability of high-performance computing architectures, scaling these methods to solve optimization problem quickly and efficiently is an important challenge.
However, one of the main bottlenecks to scaling these methods in such distributed-memory environments is the cost of communication.
Since communication cost often dominates computation cost, especially in the context of large datasets, we propose to develop efficient, scalable DCD and BDCD methods that defer communication without altering convergence behavior or solution accuracy.
Traditionally, the performance of DCD and BDCD implementations is limited by the frequency of communication.
This is due the fact that DCD and BDCD are iterative optimization methods which require communication at every iteration to train a candidate machine learning model.
Iterative algorithms and their communication bottlenecks are well known, particularly with respect to Krylov subspace methods \cite{Hoemmen_2010,Demmel_Hoemmen_Mohiyuddin_Yelick_2008,Chronopoulos_Gear_1989,Kim_Chronopoulos_1992,Carson_2015,Carson_Knight_Demmel_2013}.
In this work, we borrow ideas from Krylov subspace methods and apply them to the DCD and BDCD optimization methods to reduce the frequency of communication by a factor of $s$.
The contributions of the paper are,
\begin{itemize}
    \item Derivation and theoretical analysis of $s$-step variants of DCD and BDCD methods for solving the K-SVM and K-RR problems which reduce the frequency of communication by a factor of $s$ at the expense of additional computation and  storage.
    \item Empirical evaluation of the numerical stability (in MATLAB) of the $s$-step variants as a function of $s$ for several benchmark classification and regression datasets.
    \item Performance evaluation of the strong scaling and running time breakdown of the proposed methods on a Cray EX system, which show speedups of up to $9.8\times$ on large-scale dense and sparse datasets.
\end{itemize}

\section{Related Work}
In this section, we briefly survey the existing state of the art on scaling optimization methods with a focus on reducing the communication costs in distributed-memory and shared-memory parallel computing environments.

$s$-step methods have recently been adapted and generalized to nonlinear machine learning tasks in order to scale iterative optimization methods.
This body of work includes $s$-step variants of block coordinate descent methods applied to the ridge regression (L2-regularized least squares) problem \cite{AD_BDCD,Soori_Devarakonda_Blanco_Demmel_Gurbuzbalaban_Dehnavi_2018,Zhu_Chen_Wang_Zhu_Chen_2009}, $s$-step variants of novel stochastic FISTA (S-FISTA) and stochastic Newton (S-PNM) methods for proximal least squares \cite{Soori_Devarakonda_Blanco_Demmel_Gurbuzbalaban_Dehnavi_2018}, and 
P-packSVM which applied a variant of the $s$-step technique to obtain scalable linear SVM classifiers in distributed, cloud environments \cite{Zhu_Chen_Wang_Zhu_Chen_2009}.
These methods showed significant speedups in distributed-memory and cloud settings due to reduced communication overhead.

All of these adaptations to machine learning were build on prior $s$-step Krylov methods work from numerical linear algebra.
Prior work in numerical linear algebra \cite{chronopoulos1987class,chronopoulos91,Chronopoulos_Gear_1989,Kim_Chronopoulos_1992}
developed $s$-step Krylov methods for solving linear systems on distributed multiprocessors.
This work was further generalized to a wide variety of Krylov methods with three-term recurrences \cite{Hoemmen_2010,Mohiyuddin_Hoemmen_Demmel_Yelick_2009,Demmel_Hoemmen_Mohiyuddin_Yelick_2008}.
Further work developed practical Krylov methods with numerical stability analysis and strategies to stability $s$-step Krylov methods \cite{Carson_2015,Carson_Knight_Demmel_2013,Williams_Lijewski_Almgren_Straalen_Carson_Knight_Demmel_2014}.
Our expands this body of work to kernelized machine learning problem, particularly kernel SVM and kernel ridge regression, targeting first-order, coordinate descent methods.
One common thread in this line of work, is that $s$-step variants are mathematically equivalent to their classical counterpart.
Thus, the limiting factor on $s$ becomes numerical instability and the computation-communication trade off.

Alternative approaches to reducing communication include asynchronous methods, federated learning approaches (i.e. divide and conquer), and approximation methods.
These approaches typically relax convergence and solution accuracy requirements in order to reduce communication.
In this section, we briefly survey the asynchronous methods work.
HOGWILD! \cite{Recht_Re_Wright_Niu_2011} presents an asynchronous SGD method for shared-memory settings, which reduces the latency bottleneck by removing synchronization.
This work proves that if solution updates have bounded delay, then this asynchronous SGD variant converges to the true solution, in expectation.
This work was generalized to mini-batch SGD \cite{sallinen16}, also implemented in the shared-memory setting, shows that taking a batch size greater than $1$ allows for better memory-bandwidth utilization.
This work showed that the asynchronous mini-batch SGD method also converges and exhibits additional performance improvements over HOGWILD!
The asynchronous approach was also extended to dual coordinate descent methods for solving the SVM problem \cite{You_Lian_Liu_Yu_Dhillon_Demmel_Hsieh_2016}, where a greedy coordinate selection algorithm is introduced to further accelerate convergence of the asynchronous DCD method.
This work also targeted a multicore, shared-memory environment.

Divide-and-conquer or federated learning approaches such as CoCoA and proxCoCoA+ \cite{jaggi2014communication,SmithFJJ15} offer a framework for distributed optimization in cloud environments (using Apache Spark). 
CoCoA and proxCoCoA+, in particular, reduce synchronization by performing dual coordinate descent on locally stored data for L2 and L1-regularized convex optimization problem.
In these frameworks, each Apache Spark executor maintains a local solution which is optimized using only locally stored data.
The local solution is occasionally sum/average-reduced after a tunable number of local iterations.
This work introduces a convergence-performance trade off where differing the aggegation step for too many iterations impacts convergence and final model accuracy.
A similar approach is also shown to work for SGD \cite{Stich19} for the L2-regularized logistic regression problem.
This work sparked follow-on work in federated learning and has been generalized to many different optimization methods and machine learning problems.
The main drawback of these federated learning approaches is that they exhibit a convergence-performance trade off where scaling to additional threads/processors negatively affects convergence and solution accuracy.

Finally, we survey a subset of work on approximation methods which improve performance by exploiting the low-rank structure of the kernel matrix, for kernelized ML problems \cite{chavez20,you15-ipdps,you18-ics}.
These prior works exploit low-rank structure in two ways:
\begin{enumerate*}
    \item construct a hierarchical semi-separable (HSS) approximation to the kernel matrix using approximate nearest-neighbors \cite{chavez20}.
    \item Utilize (balanced) K-means/medoids clustering to divide data between processors \cite{you15-ipdps,you18-ics}.
\end{enumerate*}
Both approaches trade accuracy/convergence for performance in the distributed-memory setting.
\section{Derivation}
In this section, we introduce the K-SVM and K-RR problems and 
present the dual coordinate descent (DCD) and its blocked (BDCD) variant for solving these problems.
We also present $s$-step derivations of DCD for K-SVM and BDCD for K-RR.
\subsection{Support Vector Machine}
Support Vector Machine (SVM) \cite{Boser_Guyon_Vapnik_1992,Guyon_Boser_Vapnik_1992,Cortes_Vapnik_1995} are supervised learning models used for binary classification.
SVM classifies the input dataset by finding a hyperplane defined by $H := \left\{ x~|~v^\intercal x - \rho = 0 \right\}$, where $v \in \mathbb{R}^{n}$ is any vector, $\rho \in \mathbb{R}$ is the intercept, and $x \in \mathbb{R}^{n}$ is the normal vector to the hyperplane.
Given a dataset $A \in \mathbb{R}^{m \times n}$ and a vector of binary labels $y \in \mathbb{R}^m$ s.t. $y_i \in \{-1, +1\}~\forall~i=1,\ldots,m$, the SVM optimization problem finds a hyperplane which maximizes the distance (margin) between the two classes.
A good separation is achieved when the hyperplane, $x$, has the greatest distance from the nearest training data points (support vectors) from the two classes.
This formulation is known as the hard-margin SVM problem, which implicitly assumes that the data points are linearly separable.
For datasets containing errors or where the margin between the two classes is small, the hard-margin SVM formulation may not achieve high accuracy.
The soft-margin SVM problem introduces slack variables for each data point, $\xi_i~\forall~i=1,\ldots,m$, so that the margin can be increased by allowing misclassification of some data points.
Note that setting $C = 0$ recovers the hard-margin formulation, so we will focus on solving the soft-margin SVM problem:
\begin{align}
  \argmin_{x \in \mathbb{R}^n} & \frac{1}{2} \|{x}\|_{2}^2 + C\sum_{i=1}^{m} \xi_i \nonumber\\
  \text{subject to} \, \, & y_i({a}_{i,:} {x} + \rho) \leq 1 - \xi_i \nonumber\\
  & \xi_i \geq 0 \nonumber\\
  &\forall~i = 1,\ldots, m
  \label{eq:softsvm}
\end{align}

The constraints in \eqref{eq:softsvm} can be re-written in the empirical risk minimization form\footnote{We omit the bias term in the derivation for presentation clarity.} by introducing the hinge loss for each $\xi_i$:
\begin{align*}
   \textbf{SVM-L1: } &\max(1 - y_i{a_{i,:}}x, 0)\\
   \textbf{SVM-L2: } &\max(1 - y_i{a_{i,:}}x, 0)^2.
\end{align*}
We refer to the hinge loss variant as the L1-SVM problem and the squared hinge-loss variant as the L2-SVM problem.
For datasets that are not linearly separable in $n$ dimensions, the SVM problem can be solved in a high-dimensional feature space by introducing a non-linear kernel function.
The kernelized variants of the L1-SVM and L2-SVM problems can be obtained by deriving their Lagrangian dual problems,
kernelized SVM-L1 (K-SVM-L1) has the following form,
\begin{align*}
    \argmin_{\alpha \in \mathbb{R}^m} &\frac{1}{2} \sum_{i = 1}^m \sum_{j = 1}^m \alpha_i \alpha_j y_i y_j \mathcal{K}(a_{i,:},a_{j,:})  - \sum_{i = 1}^m \alpha_i\\
    &\text{subject to}~ 0 \leq \alpha_i \leq C,
\end{align*}
and kernelized SVM-L2 (K-SVM-L2) has the following form,
\begin{align*}
    \argmin_{\alpha \in \mathbb{R}^m} \frac{1}{2} \sum_{i = 1}^m \sum_{j = 1}^m \alpha_i \alpha_j y_i y_j \mathcal{K}(a_{i,:}, a_{j,:}) - \sum_{i = 1}^m \alpha_i + \frac{1}{4C}\sum_{i=1}^{m} {\alpha_i}^2.
\end{align*}
where $\alpha \in \mathbb{R}^m$ is the solution to the Lagrangian dual problem and $\mathcal{K}(a_{i,:}, a_{j,:})$ is a kernel function which defines an inner product space.
\Cref{tbl:kernels} lists the kernel functions used in this work.
\begin{table}
    \centering
    \begin{tabular}{c|c}
        Kernel function & Definition\\ \hline\hline
        \bf Linear & $a_{i,:} a_{j,:}^\intercal$\\\hline
        \bf Polynomial & $(c + a_{i,:} a_{j,:}^\intercal)^d$, for $c \geq 0, d \geq 2.$\\\hline
        \bf Radial Basis Function (RBF) & $\exp\left(-\sigma \Vert a_{i,:} - a_{j,:}\Vert_2^2\right)$, for $\sigma > 0$.
    \end{tabular}
    \caption{Kernel functions explored in this work}
    \label{tbl:kernels}
\end{table}
The L1 and L2 K-SVM problems can be solved using several algorithmic variants of coordinate descent \cite{DCD_Linear_SVM,platt1998sequential,You_Lian_Liu_Yu_Dhillon_Demmel_Hsieh_2016}.
In this work, we will focus on cyclic coordinate descent, which we will refer to as Dual Coordinate Descent (DCD).
DCD is an iterative algorithm which reduces the K-SVM problems to single-variable (or coordinate) problems which have closed-form solutions.
Once a sub-problem is solved, DCD proceeds by selecting a new variable, solving the sub-problem with respect to the chosen variable, and repeating this process until convergence.
\Cref{alg:dcd-ksvm} shows the DCD algorithm for solving the L1 and L2 K-SVM problems.
\begin{algorithm}
\caption{Dual Coordinate Descent (DCD) for Kernel SVM}
\begin{algorithmic}[1]
\STATE \textbf{Input:} $A \in \mathbb{R}^{m \times n}, y \in \mathbb{R}^m,  \alpha_0 \in \mathbb{R}^m, H > 1, C \in \mathbb{R}, \newline \mathcal{K} := A \times A \mapsto \mathbb{R}$
\STATE $\begin{cases}
    \nu = C,~\omega = 0,& \text{K-SVM-L1}\\
    \nu = \infty,~\omega = 1/2C, & \text{K-SVM-L2}
\end{cases}$
\STATE $\tilde{A} = \diag(y) \cdot A$
\FOR{$k = 1,2,\ldots,H$}
    \STATE Choose $i_k \in [m]$ uniformly at random.
    \STATE $e_{i_k} \in \mathbb{R}^m$, the $i_k$-th standard basis vector.
    \STATE ${u_k} = \mathcal{K}\left(\tilde A, e_{i_k}^\intercal \tilde A\right)$
    \STATE $\eta_k = e_{i_k}^\intercal u_k  + \omega$
    \STATE $g_k = u_k^\intercal \alpha_{k-1} - 1 + \omega e_{i_k}^\intercal \alpha_{k-1}$
    \STATE $\tilde{g}_k = \left| \min\left(\max\left(e_{i_k}^\intercal \alpha_{k-1} - g_k, 0\right), \nu\right) - e_{i_k}^\intercal \alpha_{k-1}\right|$
    \IF{$\tilde{g}_k \neq 0$}
    \STATE $\theta_k = \min\left(\max\left(e_{i_k}^\intercal \alpha_{k-1} - \frac{g_k}{\eta_k}, 0\right), \nu\right)- e_{i_k}^\intercal \alpha_{k-1}$
    \ELSE
        \STATE $\theta_k = 0$
    \ENDIF
    \STATE $\alpha_k = \alpha_{k-1} + \theta_k e_{i_k}$
\ENDFOR
\STATE \textbf{Output:} $\alpha_H$
\end{algorithmic}
\label{alg:dcd-ksvm}
\end{algorithm}
\subsection{s-Step DCD Derivation}
Note that \Cref{alg:dcd-ksvm} selects a single data point from $A$ at each iteration.
This limits DCD performance to BLAS-1 and BLAS-2 operation in each iteration.
This limitation also suggests that a distributed-memory parallel implementation of DCD would require communication at every iteration.
We propose to improve the performance of DCD by deriving a mathematically equivalent variant we refer to as $s$-step DCD which avoids communication for $s$ iterations.
We begin by modifying the iteration index from $k$ (for DCD) to $sk + j$, where $j \in \{1, 2, \ldots, s\}$ and $k \in \{0, 1, \ldots, H/s\}$.
We begin the $s$-step derivation by assuming that $\alpha_{sk}$ was just computed and show how to compute the next $s$ solution updates.
We see that from \Cref{alg:dcd-ksvm} $u_k ,g_k,$ and $\alpha_k$ are vector quantities whereas $\eta_k$ and $\theta_k$ are scalar quantities.
The following two solution updates at iterations $sk+1$ and $sk+2$ require computing the gradients $g_{sk+1}$ and $g_{sk+2}$, which are defined by:
\begin{align*}
    g_{sk+1} &= u_{sk+1}^\intercal \alpha_{sk} - 1 + \omega e_{i_{sk+1}}^\intercal \alpha_{sk} \\
    g_{sk+2} &= u_{sk+2}^\intercal \alpha_{sk+1} - 1 + \omega e_{i_{sk+2}}^\intercal \alpha_{sk+1}
\end{align*}
However, notice we can also replace $\alpha_{sk+1}$ with its equivalent quantity from iteration $sk$.
\begin{align*}
    \alpha_{sk+1} &= \alpha_{sk} + \theta_{sk+1} e_{i_{sk+1}}
\end{align*}
Given that $\theta_{sk+1}$ is a scalar quantity, $g_{sk+2}$ can be rewritten as,
\begin{align*}
    g_{sk+2} &= {u}_{sk+2}^\intercal \alpha_{sk} - 1 + \theta_{sk+1}{u}_{sk+1}^\intercal e_{i_{sk+1}} + \\
    & \omega e_{i_{sk+2}}^\intercal \alpha_{sk} + \omega e_{i_{sk+2}}^\intercal \theta_{sk+1} e_{i_{sk+1}}
\end{align*}
This recurrence unrolling suggests that $g_{sk+2}$ can be computed using $\alpha_{sk}$ provided that the quantities $u_{sk+1}, u_{sk + 2}$ and $\theta_{sk+1}$ are known.
Since $u_{sk+1}$ and $u_{sk + 2}$ are independent, they can be computed simultaneously.
The sequential dependence on $\theta_{sk+1}$, however, cannot be eliminated.
Notice that $g_{sk +j}$ is defined as
\begin{align*}
    g_{sk+j} = {u}_{sk+j}^\intercal \alpha_{sk+j-1} - 1 + \omega e_{i_{sk+j}}^\intercal \alpha_{sk+j-1}.
\end{align*}
Since $\alpha_{sk+j-1}$ can be unrolled as follows
\begin{align*}
    \alpha_{sk+j-1} &= \alpha_{sk} + \sum_{t=1}^{j-1} \theta_{sk+t}e_{sk+t},
\end{align*}
we can define $g_{sk+j}$ in terms of $\alpha_{sk}$:
\begin{align*}
    g_{sk+j} &= {u}_{sk+j}^\intercal \alpha_{sk} + {u}_{sk+j}^\intercal \sum_{t=1}^{j-1} \theta_{sk+t} e_{i_{sk+t}} \\
    &- 1 + \omega e_{i_{sk+j}}^\intercal \alpha_{sk} + \omega e_{i_{sk+j}}^\intercal \sum_{t=1}^{j-1} \theta_{sk+t} e_{i_{sk+t}}.
\end{align*}
Since all quantities $u_{sk+j}$ are independent, we can compute them upfront after which all quantities $\theta_{sk+j}$ are computed sequentially.
\Cref{alg:sstep-dcd-ksvm} shows the resulting $s$-step DCD algorithm for solving the L1 and L2 K-SVM problems.
\begin{algorithm}
\caption{$s$-Step DCD for Kernel SVM}
\label{alg:sstep-dcd-ksvm}
\begin{algorithmic}[1]
\STATE \textbf{Input:} $A \in \mathbb{R}^{m \times n}, y \in \mathbb{R}^m,  \alpha_0 \in \mathbb{R}^m, C \in \mathbb{R}, s \in \mathbb{Z}^+, H \geq s, \newline \mathcal{K} := A \times A \mapsto \mathbb{R}$
\STATE $\begin{cases}
    \nu = C,~\omega = 0,& \text{K-SVM-L1}\\
    \nu = \infty,~\omega = 1/2C, & \text{K-SVM-L2}
\end{cases}$
\STATE $\tilde{A} = \diag(y) \cdot A$
\FOR{$k = 0, \ldots,H/s$}
    \FOR{$j = 1, \ldots, s$}
        \STATE $i_{sk+j} \in [m]$, chosen uniformly at random.
        \STATE $e_{i_{sk+j}} \in \mathbb{R}^m$, the $i_{sk+j}$-th standard basis vector.
    \ENDFOR
    \STATE $V_k = \left[ e_{i_{sk+1}}, \ldots, e_{i_{sk+s}} \right]$
    \STATE $\tilde{A}_k = V_k^\intercal \tilde{A}$
    \STATE $U_k = \mathcal{K}(\tilde{A}, \tilde{A}_k)$
    \STATE $G_k = V_k^\intercal U_k + \omega I$
    \STATE $[\eta_{sk+1}, \ldots, \eta_{sk+s}]^\intercal = \text{diag}(G_k)$
    \FOR{$j = 1, \ldots, s$}
        \STATE $\rho_{sk+j} = e_{i_{sk+j}}^\intercal \alpha_{sk} + e_{i_{sk+j}}^\intercal \sum_{t=1}^{j-1} \theta_{sk+t} e_{i_{sk+t}}$
        \STATE
            $g_{sk+j} = \left({U}_{k}e_j\right)^\intercal \alpha_{sk} - 1 + \omega e_{i_{sk+j}}^\intercal \alpha_{sk} + \sum_{t=1}^{j-1}\left(e_{i_{sk+t}}^T U_ke_j\right)\theta_{sk+t} + \omega {e}_{i_{sk+j}}^\intercal \sum_{t=1}^{j-1} \theta_{sk+t} e_{i_{sk+t}}$
        \STATE $\tilde{g}_{sk+j} = \left|\min\left(\max\left(\rho_{sk+j} - g_{sk+j}, 0\right), \nu\right) - \rho_{sk+j}\right|$
        \IF{$\tilde{g}_{sk+j} \neq 0$}
            \STATE
            $\theta_{sk+j} = \min\left(\max\left(\rho_{sk+j} - \frac{g_{sk+j}}{\eta_{sk+j}}, 0\right), \nu\right) - \rho_{sk+j}$
        \ELSE
            \STATE $\theta_{sk+j} = 0$
        \ENDIF
    \ENDFOR
    \STATE $\alpha_{sk+s} = \alpha_{sk} + \sum_{t=1}^s \theta_{sk+t}e_{i_{sk+t}}$
\ENDFOR
\STATE \textbf{Output:} $\alpha_H$
\end{algorithmic}
\end{algorithm}
\subsection{Kernel Ridge Regression}
The ridge regression problem \cite{kernel-ridge} for regression is defined as follows:
$
    \argmin_{x \in \mathbb{R}^n} \frac{1}{2m} \Vert Ax - y \Vert_2^2 + \frac{\lambda}{2} \Vert x \Vert_2^2.
$
By deriving the Lagrangian dual formulation, we obtain the Kernel Ridge Regression (K-RR) problem,
\begin{equation}
    \argmin_{\alpha \in \mathbb{R}^m} \frac{1}{2} \left(\sum_{i = 1}^m \sum_{j = 1}^m \alpha_i \alpha_j \left(\frac{1}{\lambda}\mathcal{K}(a_{i,:}, a_{j,:}) + m\right)\right) - \sum_{i = 1}^m \alpha_i y_i
    \label{eq:krr-problem}
\end{equation}
In contrast to K-SVM, this problem can be solved in closed form.
However, when $m \gg 1$, explicitly computing and storing the $m \times m$ kernel matrix can be prohibitive.
In this work, we focus on iteratively solving the K-RR problem using the Block Dual Coordinate Descent (BDCD) method.
BDCD solves \eqref{eq:krr-problem} by randomly or cyclically selecting a block size, $b$, of samples from $A$ and solving a subproblem with respect to just those coordinates.
\Cref{alg:bdcd} shows the BDCD algorithm for solving the K-RR problem.
\begin{algorithm}
\caption{Block Dual Coordinate Descent (BDCD) for Kernel Ridge Regression}
\label{alg:bdcd}
\begin{algorithmic}[1]
\STATE \textbf{Input:} $A \in \mathbb{R}^{m \times n}, y \in \mathbb{R}^m, b \in \mathbb{Z}^+  s.t. b \leq m, H > 1, \alpha_0 \in \mathbb{R}^m$
\STATE $\mathcal{K} := A \times A \mapsto \mathbb{R}$
\FOR{\( k = 0,\ldots,H\)}
    \STATE choose \( \{i_l \in [m] | l = 1, 2, \ldots, b\} \) uniformly at random without replacement
    \STATE \( V_k = [e_{i_1}, e_{i_2}, \ldots, e_{i_{b}}] \)
    \STATE {$U_k = \mathcal{K}({A, V_k}^\intercal A)$}
    \STATE \( G_k = \frac{1}{\lambda} V_k^\intercal U_k + m I \)
    \STATE \(
    \begin{aligned}
        \Delta \alpha_k = G^{-1}_k &\left(V_k^\intercal y - m V_k^\intercal \alpha_{k-1} - \frac{1}{\lambda}U_k^\intercal \alpha_{k-1}\right)
    \end{aligned}
    \)
    \STATE \( \alpha_k = \alpha_{k-1} + V_k \Delta \alpha_k \)
\ENDFOR
\STATE \textbf{Output} \( \alpha_{H} \)
\end{algorithmic}
\end{algorithm}
\subsection{s-Step BDCD Derivation}
Similar to DCD, we can also unroll the recurrence relationship in \Cref{alg:bdcd}.
We see that from \Cref{alg:bdcd} that the matrix quantities $U_k \in \mathbb{R}^{b \times m}$ and $G_k \in \mathbb{R}^{b \times b}$ are required at every iteration to solve the subproblem and compute the vector quantity $\Delta \alpha_k \in \mathbb{R}^b$.
We begin the $s$-step derivation by modifying the iteration index from $k$ to $sk + j$ where $j \in \{1, 2, \ldots, s\}$ and $k \in \{0, 1, \ldots, H/s\}$, as before.
We assume that $\alpha_{sk}$ was just computed and show how to compute the next $s$ solution updates.
We focus on the subproblem solutions defined at iterations $sk +1$ and $sk+2$ for the $s$-step derivation, which are given by,
\begin{align*}
    \Delta \alpha_{sk+1} &= G^{-1}_{sk+1} \left(V_{sk+1}^\intercal y - m V_{sk+1}^\intercal \alpha_{sk} - \frac{1}{\lambda}U_{sk+1}^\intercal \alpha_{sk}\right)\\
    \Delta \alpha_{sk+2} &= G^{-1}_{sk+2} \left(V_{sk+2}^\intercal y - m V_{sk+2}^\intercal \alpha_{sk+1} - \frac{1}{\lambda}U_{sk+2}^\intercal \alpha_{sk+1}\right).
\end{align*}
Using the solution update, $\alpha_{sk+1} = \alpha_{sk} + V_{sk+1}\Delta \alpha_{sk+1}$, we can unroll the recurrence for $\Delta \alpha_{sk+2}$ by substitution.
This yields
\begin{align*}
    \Delta \alpha_{sk+2} = G^{-1}_{sk+2} \biggr(V_{sk+2}^\intercal y &- m V_{sk+2}^\intercal \alpha_{sk} -m V_{sk+2}^\intercal V_{sk+1}\Delta \alpha_{sk+1} \\
    &- \frac{1}{\lambda}U_{sk+2}^\intercal \alpha_{sk} - \frac{1}{\lambda}U_{sk+2}^\intercal V_{sk+1}\Delta \alpha_{sk+1}\biggr).
\end{align*}
Notice that the solutions at iteration $sk+1$ and $sk+2$ both depend on $\alpha_{sk}$, but $\Delta \alpha_{sk+2}$ requires additional correction terms because $\alpha_{sk+1}$ is never explicitly formed.
This recurrence unrolling can be generalized to an arbitrary future iteration, $sk + j$, as follows
\begin{align}
    \Delta \alpha_{sk+j} = G^{-1}_{sk+j} \biggr(V_{sk+j}^\intercal y &- m V_{sk+j}^\intercal \alpha_{sk} -m \sum_{t = 1}^{j-1} V_{sk+j}^\intercal V_{sk+t}\Delta \alpha_{sk+t} \nonumber \\
    &- \frac{1}{\lambda}U_{sk+j}^\intercal\alpha_{sk} - \frac{1}{\lambda}\sum_{t = 1}^{j-1}U_{sk+j}^\intercal V_{sk+t}\Delta \alpha_{sk+t}\biggr). \label{eq:bdcd-sstep}
\end{align}
In \eqref{eq:bdcd-sstep} the quantities $U_{sk + j}$ for $j  \in \{1, 2, \ldots, s\}$ can be computed upfront by selecting $sb$ coordinates and computing a kernel matrix that has a factor of $s$ additional rows.
The sequence of $U_{sk+j}$'s can then be extracted from the $m \times sb$ kernel matrix.
The resulting $s$-step BDCD algorithm for K-RR is shown in \Cref{alg:cabdcd}.
\begin{algorithm}
\caption{$s$-Step BDCD for Kernel Ridge Regression}
\label{alg:cabdcd}
\begin{algorithmic}[1]
\STATE \textbf{Input:} $A \in \mathbb{R}^{m \times n}, y \in \mathbb{R}^m, b \in \mathbb{Z}^+  s.t. b \leq m,$
\STATE $ s \in \mathbb{Z}^+, H \geq s, \alpha_0 \in \mathbb{R}^m,\mathcal{K} := A \times A \mapsto \mathbb{R}$
\FOR{\( k = 0, \ldots H/s, \)}
    \FOR{\( j = 1, 2, \ldots, s \)}
        \STATE Choose \( \{ i_l \in [m] \mid l = 1, 2, \ldots, b \} \) uniformly
        \STATE at random without replacement
        \STATE \( V_{sk+j} = \left[ e_{i_1}, e_{i_2}, \ldots, e_{i_{b}} \right] \in \mathbb{R}^{m \times b} \)
    \ENDFOR
    \STATE $\Omega_k = [V_{sk +1}, V_{sk + 2}, \ldots, V_{sk + s}] \in \mathbb{R}^{m \times sb}$
    \STATE $Q_k = \mathcal{K}( A, \Omega_k^\intercal A) \in \mathbb{R}^{m \times sb}$
    \FOR{\( j = 1, 2, \ldots, s \)}
        \STATE $\Psi = [e_{jb-b + 1}, \ldots e_{jb}] \in \mathbb{R}^{sb \times b}$
        \STATE $U_{sk+j} = Q_k \Psi$
         \STATE $G_{sk+j} = \frac{1}{\lambda}V_{sk+j}^\intercal U_{sk+j} + m I$
            \begin{align*}
                \Delta \alpha_{sk+j} &= G^{-1}_{sk+j} \biggr(V_{sk+j}^\intercal y - m V_{sk+j}^\intercal \alpha_{sk}\\
                 &- m \sum_{t = 1}^{j-1} V_{sk+j}^\intercal V_{sk+t}\Delta \alpha_{sk+t} \\
    		&- \frac{1}{\lambda}U_{sk+j}^\intercal \alpha_{sk} - \frac{1}{\lambda}\sum_{t = 1}^{j-1}U_{sk+j}^\intercal V_{sk+t}\Delta \alpha_{sk+t}\biggr)
            \end{align*}
    \ENDFOR
    \STATE \( \alpha_{sk+s} = \alpha_{sk} + \sum_{t=1}^s V_{sk+t} \Delta \alpha_{sk+t} \)
\ENDFOR
\STATE \textbf{Output} \( \alpha_{H} \)
\end{algorithmic}
\end{algorithm}
\section{Parallel Algorithms and Analysis}
In this section we analyze their computation and communication costs (\Cref{subsec:analysis}) using Hockney's performance model \cite{hockney}:
    $\gamma F + \beta W + \phi L,$
where $F, W$ and $L$ represent the algorithm costs for computation, bandwidth, and latency; and $\gamma, \beta$ and $\phi$ represent the associated hardware parameters.
Given the similarities between the coordinate descent methods for K-RR and K-SVM, we focus our analysis on \Cref{alg:bdcd,alg:cabdcd} which are blocked generalizations of DCD.
The leading-order costs of BDCD for K-RR can be specialized to DCD for K-SVM by setting the block size, $b = 1$.
Independent analyses of BDCD and DCD are required in order to bound constants, which we omit in this work.
\subsection{Computation and Communication Analysis}\label{subsec:analysis}
We assume that $A \in \mathbb{R}^{m \times n}$ is a sparse matrix with density $f$, where $0 < f \leq 1$, such that the non-zeros are uniformly distributed where each row contains $fn$ non-zero entries.
We further assume that the processors are load balanced with $fmn/P$ non-zeros per processor.
The K-SVM and K-RR problems require non-linear kernel operations which are often more expensive than floating-point arithmetic.
For example, the polynomial kernel requires a pointwise \texttt{pow} instruction for each entry of the sampled kernel matrix.
Similarly, the RBF kernel requires an \texttt{exp} instruction for each entry.
We model this overhead by introducing a scalar, $\mu$, to represent the cost of applying a non-linear function relative to floating-point multiplies during the kernel computation.
Note that due to the non-linear kernel function (particularly the RBF kernel), we store the $m \times b$ (sampled) kernel matrix in dense format.
Finally, the BDCD and $s$-step BDCD algorithms require an MPI Allreduce call at each iteration.
We assume that the Allreduce has a cost of $L = O(\log P)$ and $W = O(w)$ where $w$ is the message size (in words) \cite{mpicost}.
We begin by proving the parallel computation, communication, and storage costs of the BDCD algorithm (\Cref{thm:bdcd}) followed by analysis of the $s$-step BDCD algorithm (\Cref{thm:sstepbdcd}).
\begin{theorem}
    \label{thm:bdcd}
    Let \( H \) be the number of iterations of the Block Dual Coordinate Descent (BDCD) algorithm, $b$ the block size, $P$ the number of processors, and \( A \in \mathbb{R}^{m \times n} \) the matrix that is partitioned using 1D-column layout. Under this setting, BDCD has the following asymptotic costs along the critical path:
    \begin{align*}
        \textbf{Computation: } & \mathcal{O}\left(H \left(\frac{bfmn}{P} + \mu bm + b^3\right)\right) \text{ flops, }\\
        \textbf{Bandwidth: } & \mathcal{O}\left(H bm\right) \text{ words moved, }\\
        \textbf{Latency: } & \mathcal{O}\left(H \log P\right) \text{ messages,}\\
        \textbf{Storage: } & \mathcal{O}\left(bm + \frac{fmn}{P}\right) \text{ words of memory.}
    \end{align*}
\end{theorem}
\begin{proof}
    Each iteration of parallel BDCD for K-RR begins by computing $U_k = \mathcal{K}(A, V_k^\intercal A)$, which costs at most $\frac{bfmn}{P}$ flops, in parallel, to form the $m \times b$ sampled columns, $AA^TV_k$, of the $m \times m$ full kernel matrix, $AA^T$.
    Each processor forms the $m \times b$ partial kernel matrix which must be sum-reduced before applying the non-linear kernel operation.
    This requires $bm$ words to be moved and $\log P$ messages.
    After communication, the non-linear kernel operation can be applied to each entry of the kernel matrix using $\mu bm$ flops.
    The quantity $G_k$ can be extracted directly from $U_k$ by selecting the $b$ rows, corresponding to $V_k^\intercal$, sampled for this iteration.
    Since each processor redundantly stores $y$ and $\alpha$, the vector quantity $V_k^\intercal y - m V_k^\intercal - \frac{1}{\lambda}U_k^\intercal\alpha_{k-1}$ can be computed independently in parallel on each processor.
    This operation requires $bm$ flops.
    Once the vector quantity is computed, the linear system can be solved in $b^3$ flops to compute $\Delta \alpha_k$ redundantly on each processor.
    Finally, $\alpha_k$ can be updated using $b$ flops by updating only the coordinates of $\alpha$ sampled in this iteration.
    Summing and multiplying the above costs by $H$, the total number of BDCD iteration, proves the computational and communication costs.
    The storage costs can be obtained by noticing that each iteration of BDCD must simultaneously store $A$ and $U_k$ in memory, which requires $\frac{fmn}{P}$ words and $bm$ respectively.
    Each iteration also requires $G_k$, which requires $b^2$ words of storage.
    However, storing $G_k$ and other vector quantities are low-order terms in comparison to storing $A$ and $U_k$.
\end{proof}
\begin{theorem}
    \label{thm:sstepbdcd}
    Let \( H \) be the number of iterations of the $s$-Step Block Dual Coordinate Descent ($s$-Step BDCD) algorithm, $b$ the block size, $P$ the number of processors, and \( A \in \mathbb{R}^{m \times n} \) the matrix that is partitioned using 1D-column layout. Under this setting, $s$-Step BDCD has the following asymptotic costs along the critical path:
    \begin{align*}
        \textbf{Computation: } & \mathcal{O}\left(\frac{H}{s} \left(\frac{sbmfn}{P} + \mu sbm + sb^3 + \binom{s}{2}b^2\right) \right) \text{ flops, }\\
        \textbf{Bandwidth: } & \mathcal{O}\left(\frac{H}{s} sbm\right) \text{ words moved, }\\
        \textbf{Latency: } & \mathcal{O}\left(\frac{H}{s} \log P\right) \text{ messages, }\\
        \textbf{Storage: } & \mathcal{O}\left(\frac{fmn}{P} + sbm\right) \text{ words of memory.}
    \end{align*}
\end{theorem}
\begin{proof}
    The $s$-Step BDCD algorithm for K-RR computes a factor of $s$ larger kernel matrix, $Q_k = \mathcal{K}(A, \Omega_kA)$, which costs at most $\frac{sbfmn}{P}$ flops, in parallel, to partially form.
    The partial $Q_k$'s on each processor must be sum-reduced prior to applying the non-linear kernel operation, which requires $sbm$ words to be moved and $\log P$ messages.
    Once $Q_k$ is sum-reduced the non-linear kernel computation can be performed redundantly on each processor.
    Since the non-linear operation is required for each entry of $Q_k \in \mathbb{R}^{m \times sb}$, forming the sampled kernel matrix costs $\mu sbm$ flops.
    As with parallel BDCD, the $b \times b$ matrices $G_{sk+j}$ for $j \in \{1,2,\ldots,s\}$ can be extracted from $Q_k$ and is a low-order cost.
    However, the $s$-Step BDCD algorithm requires additional matrix-vector computations to form the right-hand side of the linear system.
    Forming the right-hand side requires a total of $s$ $U_{sk+j}^\intercal \alpha_{sk}$ matrix-vector computations which costs $sbm$ flops and a sequence of matrix-vector computations, $U_{sk+j}^\intercal V_{sk+t}\Delta \alpha_{sk+t}$ for $t \in \{1, \ldots, j-1\}$, to correct the right-hand side.
    There are a total of $\binom{s}{2}$ such corrections with each requiring $b^2$ flops.
    Once the corrections have been performed, $s$ linear systems can be solved using $s b^3$ flops.
    The leading-order cost of the inner loop (indexed by $j$) is $\binom{s}{2}b^2 + sbm + sb^3$ flops.
    Once the sequence of $s$ $\Delta \alpha_{sk + j}$ vectors have been computed, $\alpha_{sk+s}$ can be computed by summing with appropriate indices of $\alpha_{sk}$, which requires $sb$ flops.
    The $s$-Step BDCD computes $s$ solutions every outer iteration, so a total of $H/s$ outer iterations are required to perform the equivalent of $H$ BDCD iterations.
    Multiplying these costs by $H/s$ proves the computation and communication costs.
    Finally, this algorithm requires storage of $A, Q_k, G_{sk+j}$, and several vector quantities.
    However, the leading order costs are the storage of $A$ and $Q_k$ which costs $\frac{fmn}{P} + sbm$ words.
\end{proof}
\section{Experiments}
\begin{table}[t]
    \centering
    \begin{tabular}{|c|c|c|c|c|}
        \hline
        Name & Type & $m$ & $n$ \\
        \hline
        duke breast-cancer & binary classification & 44 & 7129  \\
        \hline
        diabetes & binary classification & 768 & 8 \\
        \hline
        abalone & regression & 4177 & 8  \\
        \hline
        bodyfat & regression & 252 & 14 \\
        \hline
    \end{tabular}
    \caption{Datasets used in the convergence experiments.}
    \label{table:matlab_datasets}
\end{table}
\begin{figure*}
    \centering
    \setkeys{Gin}{width=1\linewidth}
    \begin{subfigure}[t]{0.3\textwidth}
      \includegraphics{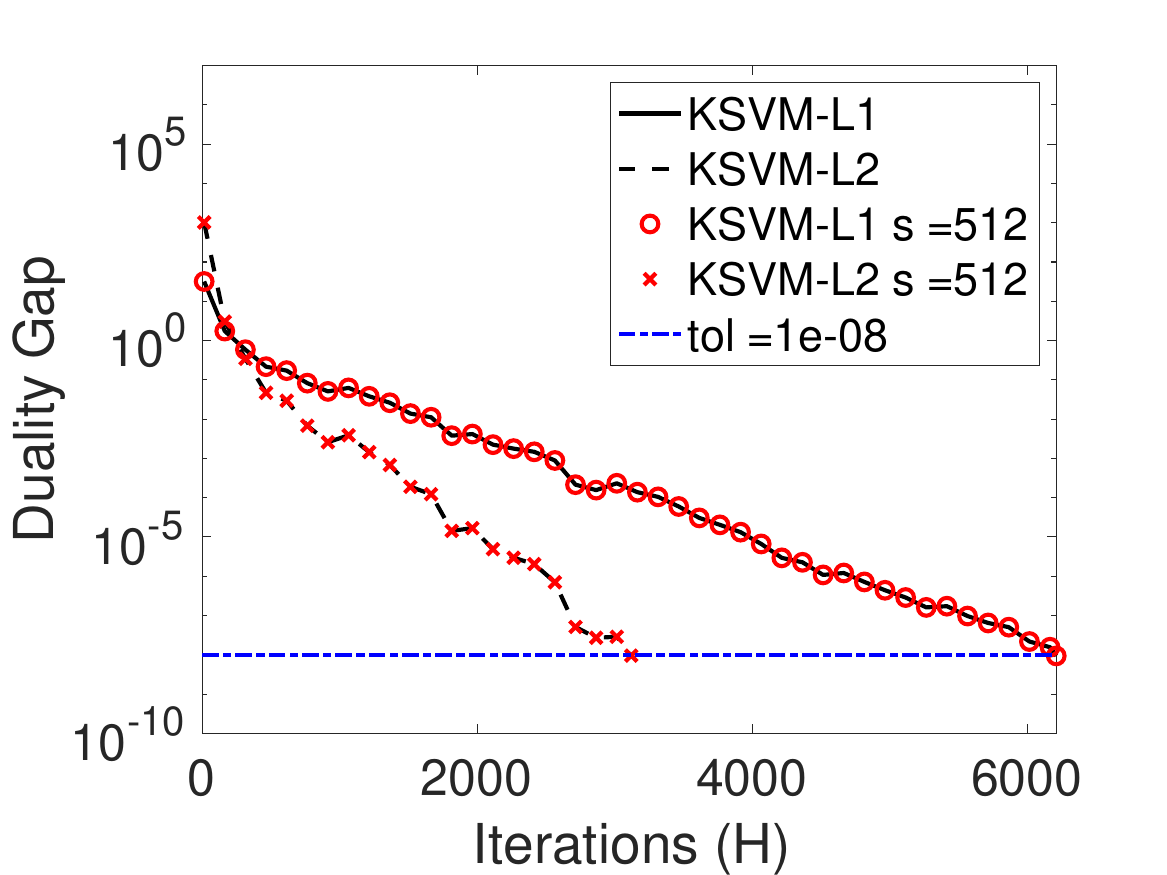}
      \caption{duke, linear}
    \end{subfigure}\hfill
    \begin{subfigure}[t]{0.3\textwidth}
      \includegraphics{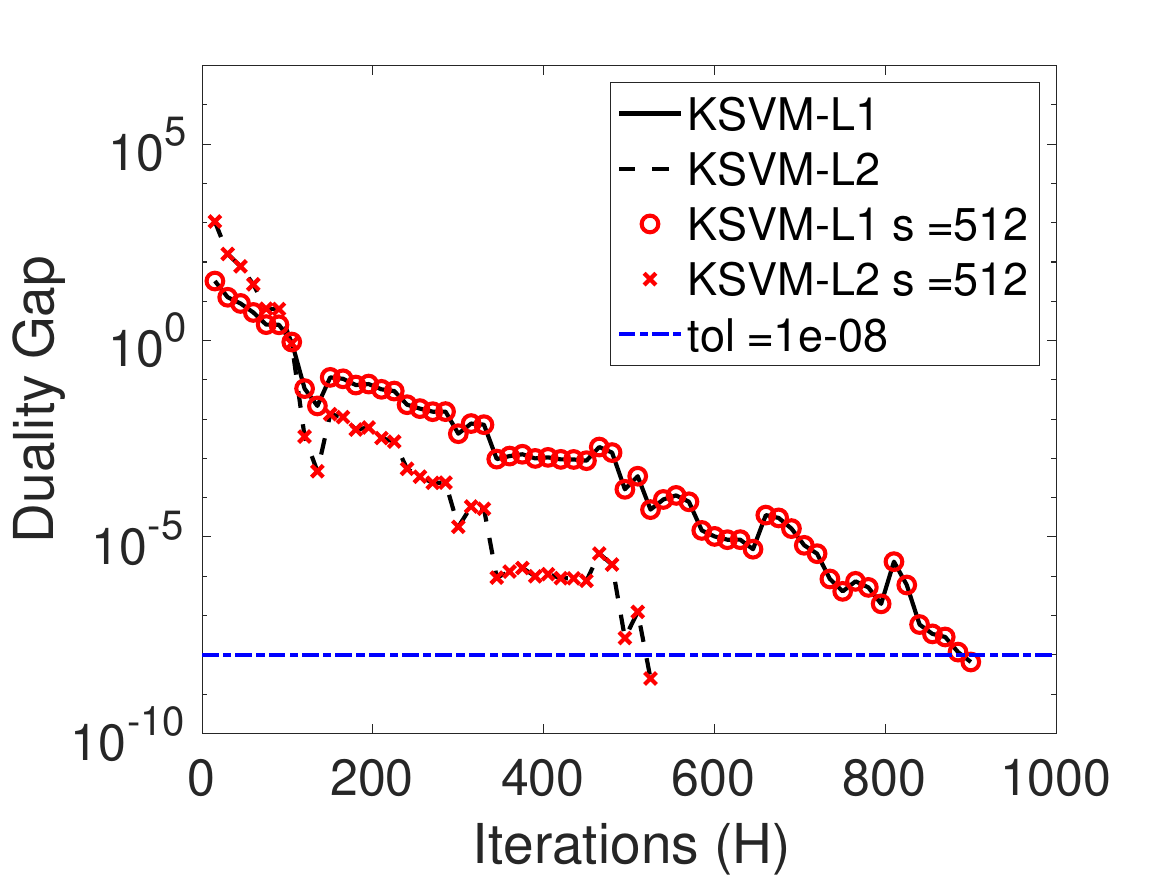}
      \caption{duke, poly}
    \end{subfigure}\hfill
    \begin{subfigure}[t]{0.3\textwidth}
      \includegraphics{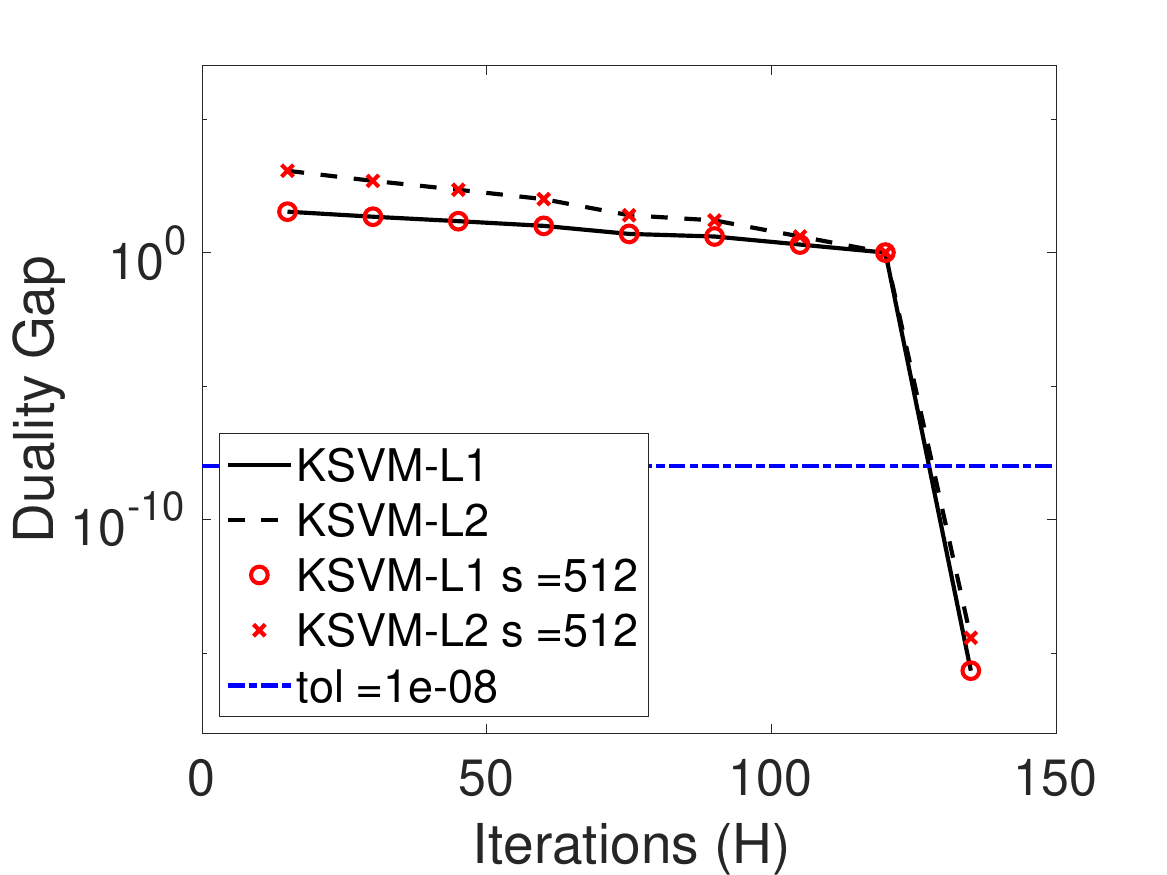}
      \caption{duke, gauss}
    \end{subfigure}\hfill
    \begin{subfigure}[t]{0.3\textwidth}
      \includegraphics{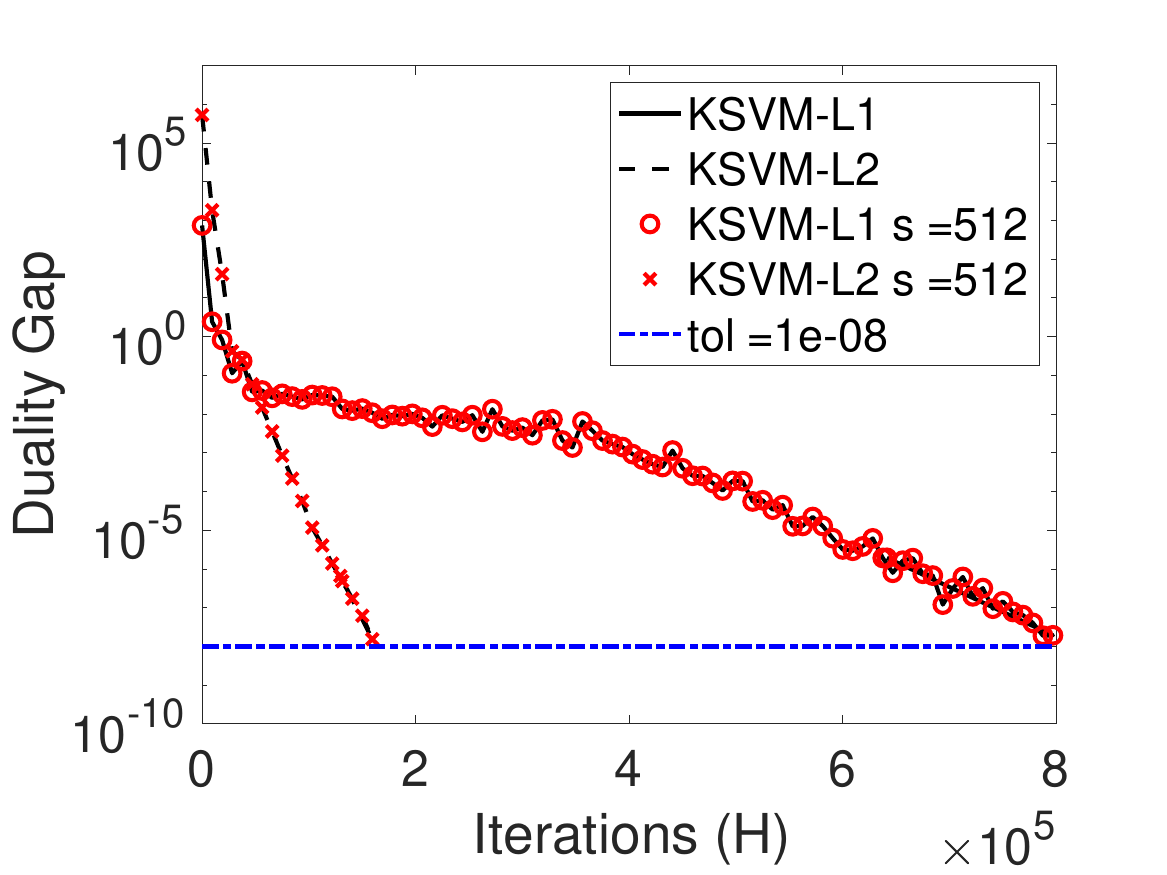}
      \caption{diabetes, linear}
    \end{subfigure}\hfill
    \begin{subfigure}[t]{0.3\textwidth}
      \includegraphics{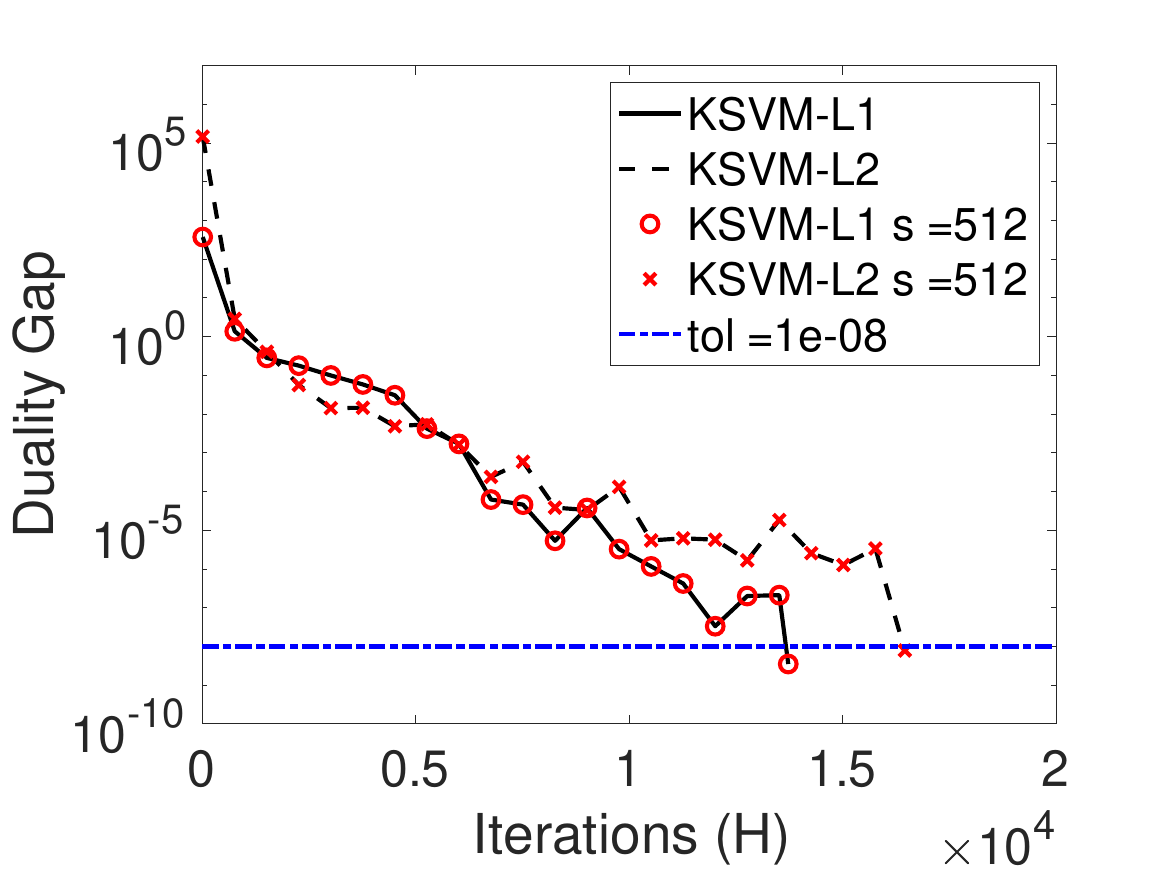}
      \caption{diabetes, poly}
    \end{subfigure}\hfill
    \begin{subfigure}[t]{0.3\textwidth}
      \includegraphics{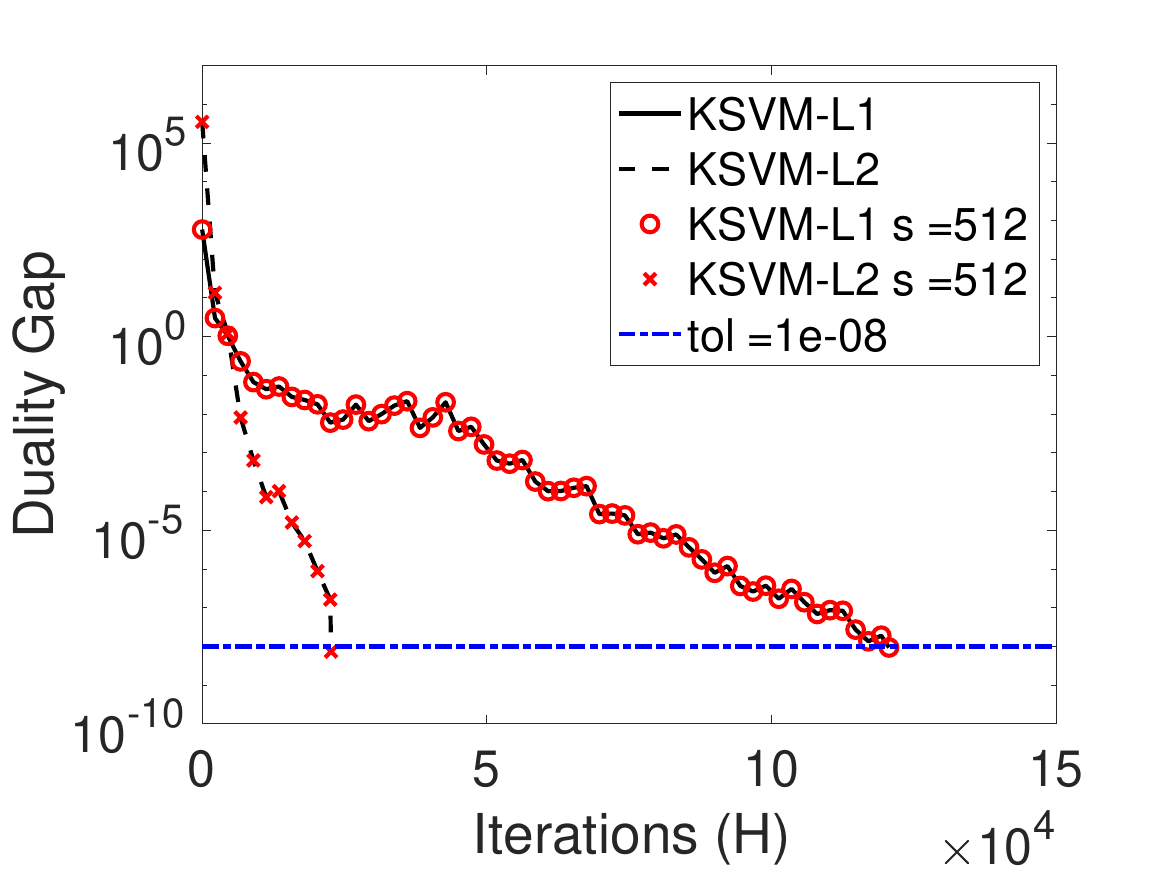}
      \caption{diabetes, gauss}
    \end{subfigure}
    \caption{Comparison of DCD and $s$-step DCD convergence behavior for K-SVM-L1 and K-SVM-L2 problems.}
    \label{fig:convergence-ksvm}
  \end{figure*}
This section presents the numerical and performance experiments of the $s$-step DCD and BDCD algorithms for K-SVM and K-RR, respectively.
Prior work on $s$-step Krylov methods \cite{Hoemmen_2010,Carson_2015} showed that the additional computation led to numerical instability due to the Krylov basis becoming rank-deficient.
The $s$-step DCD and BDCD algorithms require computation of factor of $s$ larger kernel matrix and requires gradient correction due to the deferred update on $\alpha$.
We begin by studying the convergence behavior of the $s$-step methods relative to the classical DCD and BDCD methods on several binary classification and regression datasets.
Then, we show the performance trade off and scaling behavior of the high-performance, distributed-memory implementations (written in C and MPI) of the $s$-Step DCD and BDCD methods relative to the standard DCD and BDCD methods on a Cray EX cluster.
\subsection{Convergence Experiments}\label{subsec:convergence}
We measure the convergence behavior of the $s$-step methods against the classical methods and perform ablation studies by varying the values of $s$, kernel choice (shown in \Cref{tbl:kernels}), and the block size (for the K-RR problem).
The datasets used in the experiments were obtained from the LIBSVM repository \cite{libsvm} whose properties are shown in \Cref{table:matlab_datasets}.
All algorithms were implemented and tested in MATLAB R2022b Update 7 on a MacBook Air 2020 equipped with an Apple M1 chip.
We measure the convergence of the K-SVM problem by plotting the duality gap for each setting of the $s$-step DCD method in comparison to the classical DCD method.
Since K-SVM is a convex problem, we should expect the duality gap to approach machine precision.
However, in these experiments we set the duality gap tolerance to $10^{-8}$.
We plot the duality gap over $H$ iterations until the gap converges to the specified tolerance.
Duality gap is defined as follows, $D(x) - P(x)$,
where the $D(x)$ refers to the objective value of the K-SVM problem (Lagrangian dual problem) and $P(x)$ which refers to the objective value of the primal K-SVM problem as computed by LIBSVM \cite{DCD_Linear_SVM,libsvm}.
The primal and dual SVM problems are defined in \eqref{eq:softsvm} and \eqref{eq:krr-problem}, respectively.
Note that we report convergence for both the K-SVM-L1 and (smoothed) K-SVM-L2 problems.
\Cref{fig:convergence-ksvm} shows the convergence behavior (in terms of duality gap) of the DCD and $s$-step DCD methods on the datasets in \Cref{table:matlab_datasets} and the kernels described in \Cref{tbl:kernels}.
We  polynomial kernel utilizes a degree $d = 3$ and $c = 0$.
The RBF kernel utilizes $\sigma = 1$.
We can observe for both datasets and all kernels, that the $s$-step DCD methods for K-SVM-L1 and K-SVM-L2 (red markers) exhibit the same convergence behavior as the DCD methods and attain the same solution, $\alpha_H$, up to machine precision.
Note that we expect the convergence behavior of K-SVM-L1 and K-SVM-L2 to differ since they solve different problems and attain different solutions.
We also perform experiments for the K-RR problem shown in \eqref{eq:krr-problem}.
Since K-RR has a closed-form solution, we use relative solution error to illustrate convergence behavior of BDCD and $s$-step BDCD.
The relative solution error is defined as follows,
${\Vert \alpha_k - \alpha^*\Vert_2}/{\Vert \alpha^* \Vert_2}$,
where $\alpha_k$ is the partial solution at iteration $k$ for BDCD (or iteration $sk + j$ for $s$-step BDCD) and $\alpha^*$ is the optimal solution obtained via matrix factorization.
In order to compute $\alpha^*$ we first compute the full, $m \times m$ kernel matrix and solve the linear system for $\alpha^*$.
Figure 2 compares the convergence behavior of BDCD and CA-BDCD for $s < 1$ with all three kernels.
\begin{figure*}
  \centering
  \setkeys{Gin}{width=1\linewidth}
  \begin{subfigure}[t]{0.27\textwidth}
    \includegraphics {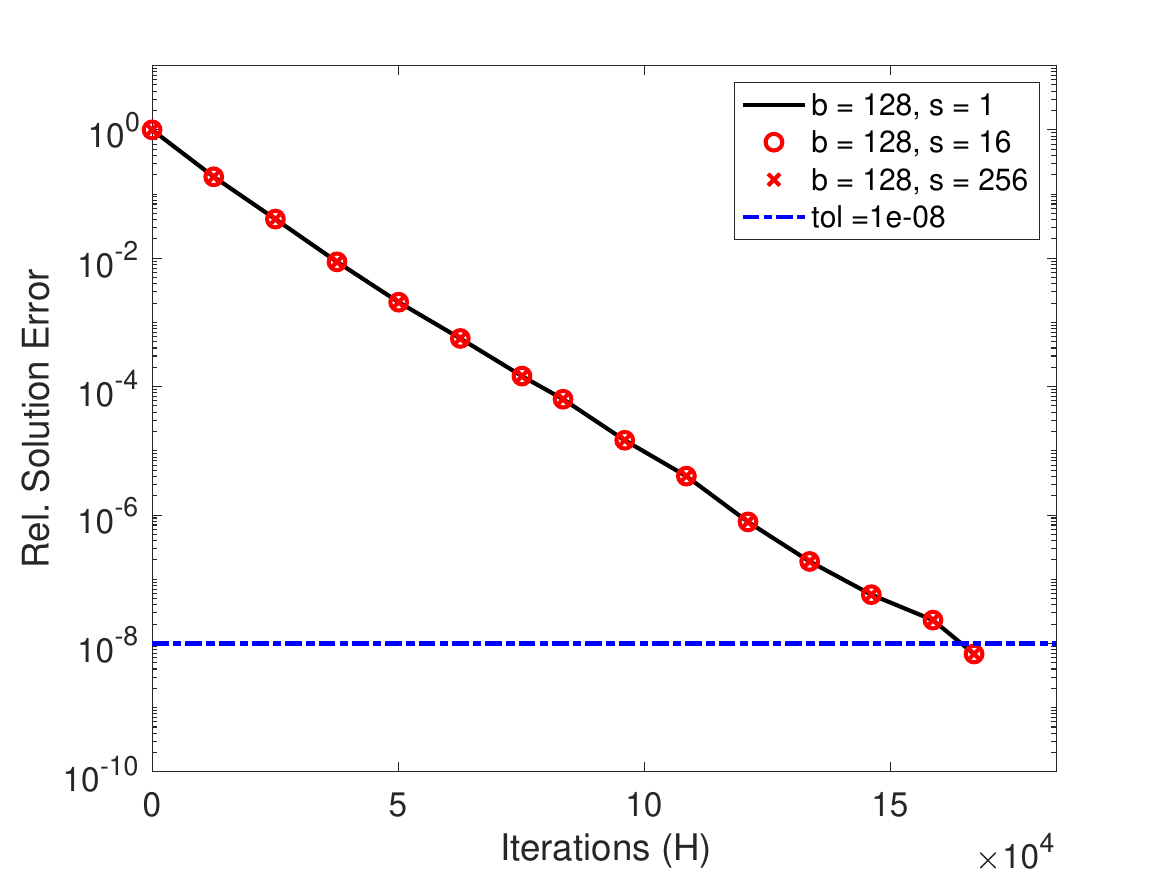}
    \caption{abalone, linear}
  \end{subfigure}\hfill
  \begin{subfigure}[t]{0.27\textwidth}
    \includegraphics {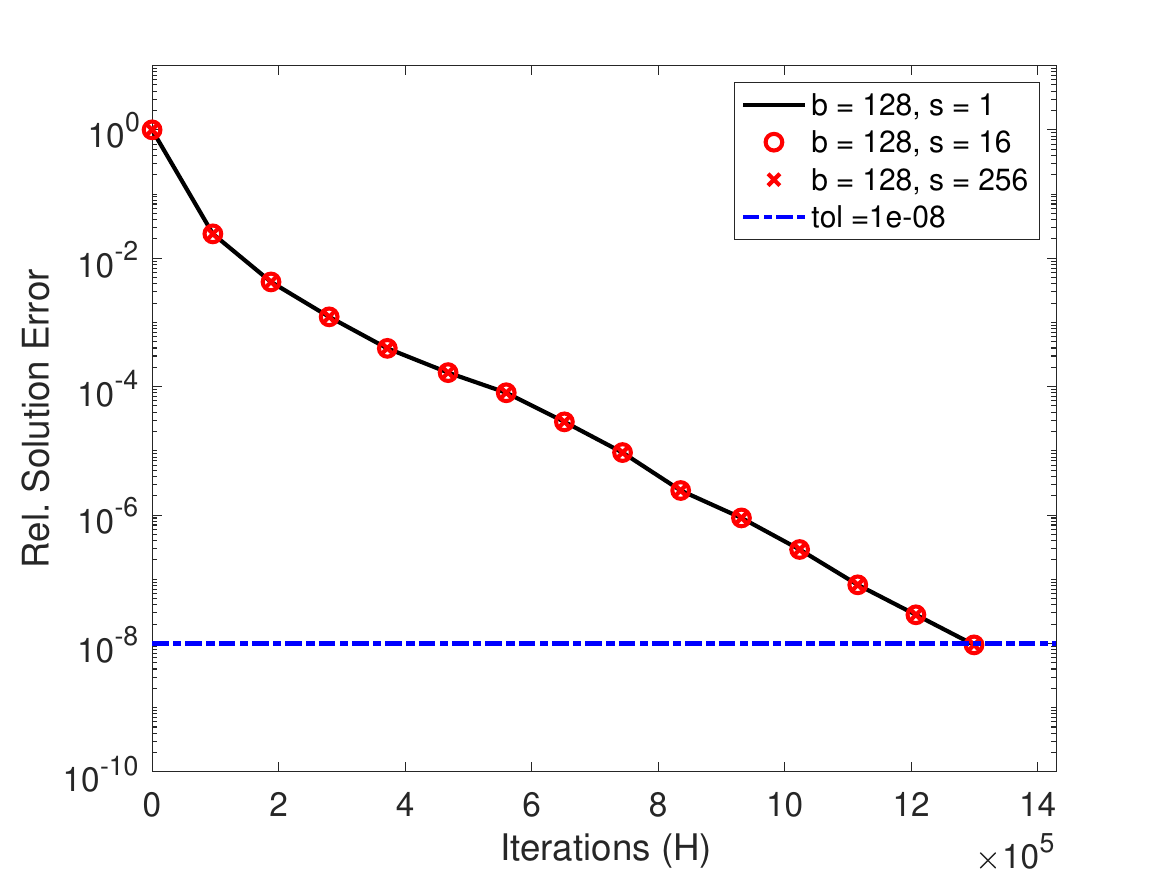}
    \caption{abalone, poly}
  \end{subfigure}\hfill
  \begin{subfigure}[t]{0.27\textwidth}
    \includegraphics {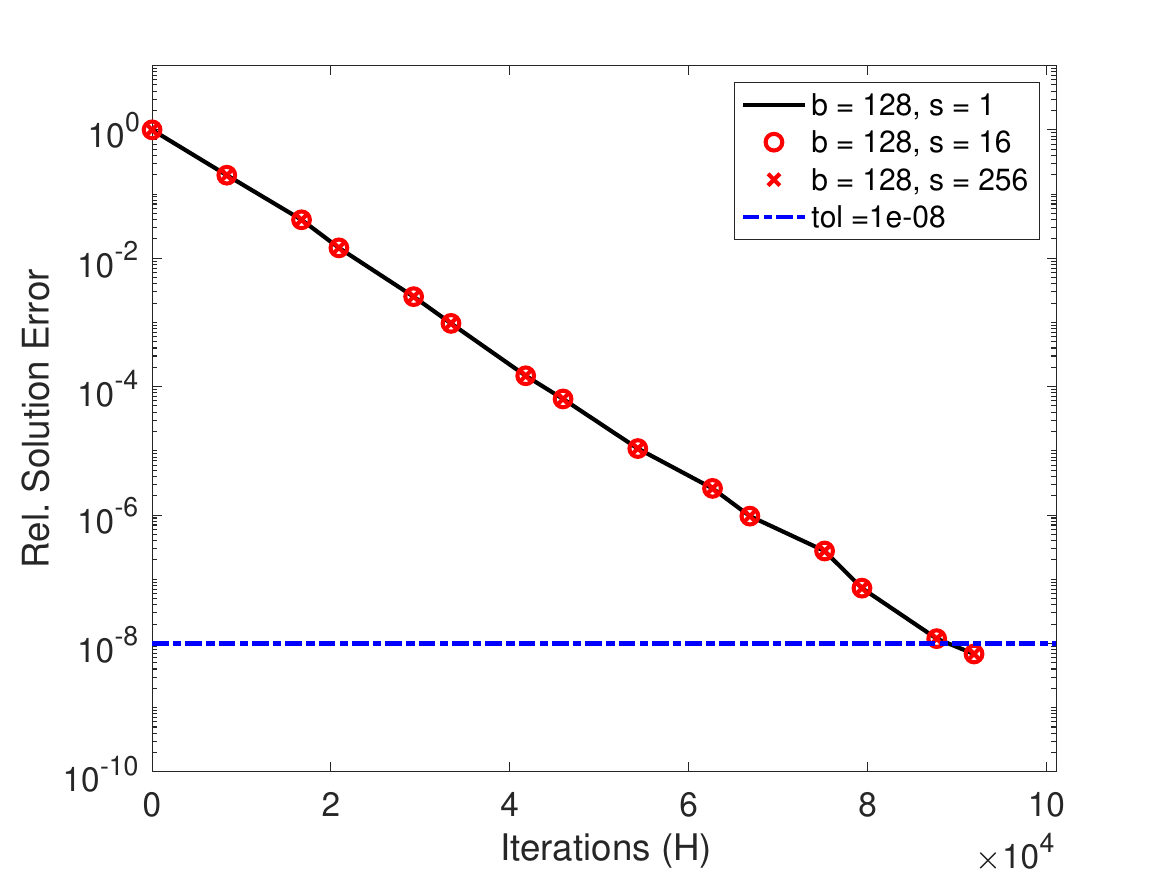}
    \caption{abalone, gauss}
  \end{subfigure}\hfill
  \begin{subfigure}[t]{0.27\textwidth}
    \includegraphics {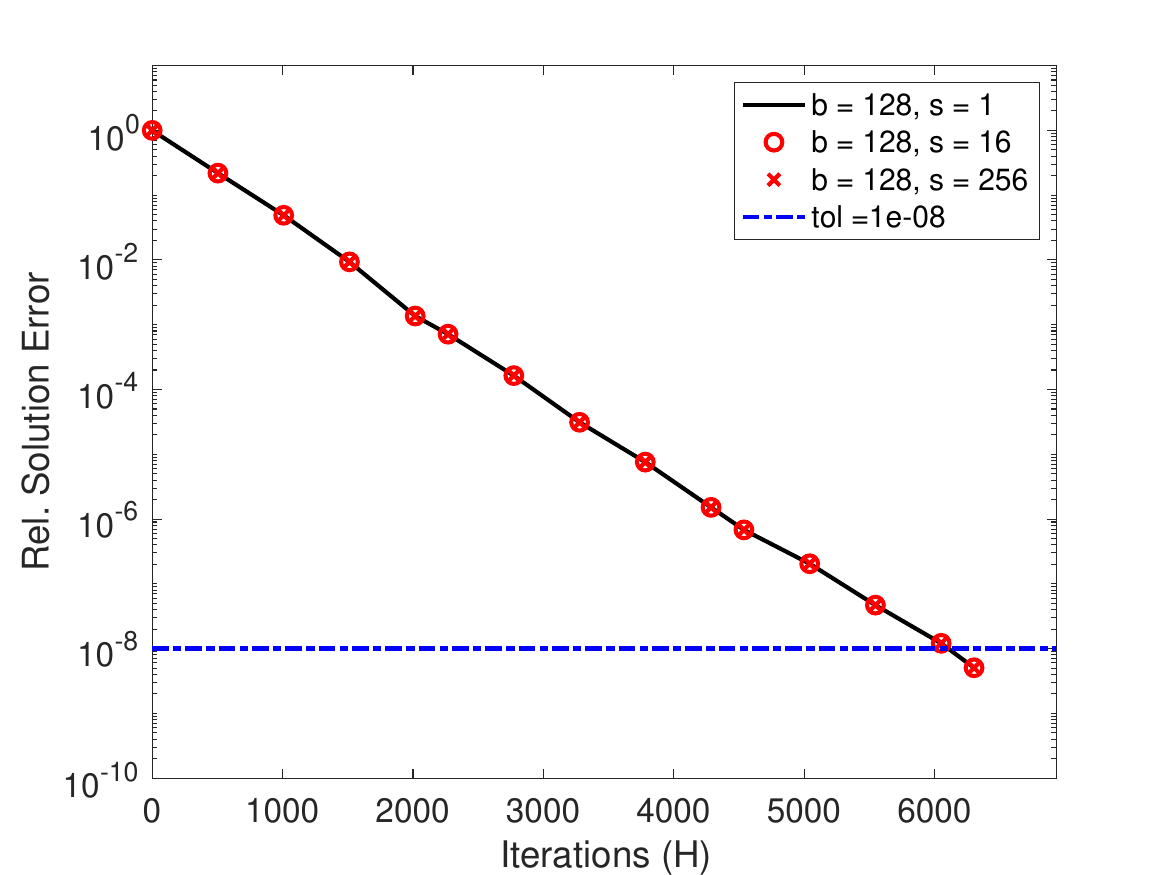}
    \caption{bodyfat, linear}
  \end{subfigure}\hfill
  \begin{subfigure}[t]{0.27\textwidth}
    \includegraphics {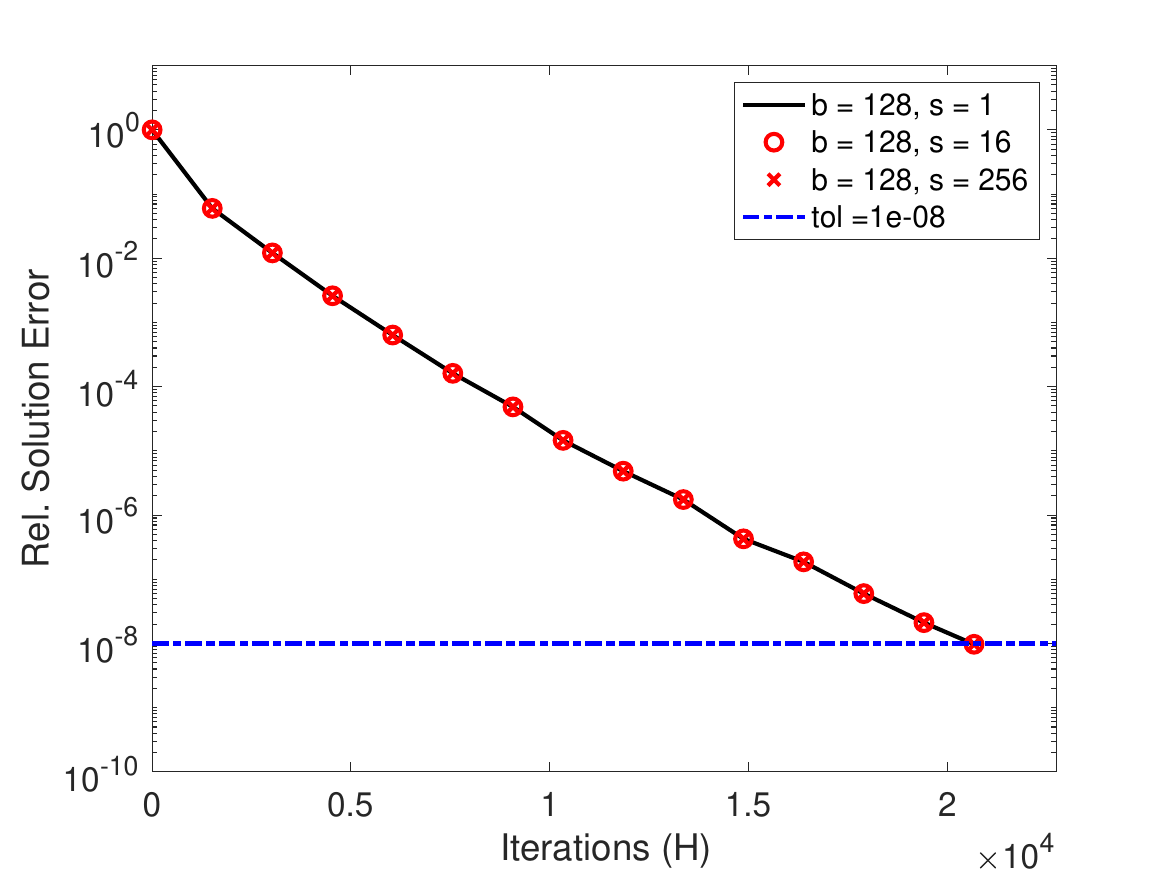}
    \caption{bodyfat, poly}
  \end{subfigure}\hfill
  \begin{subfigure}[t]{0.27\textwidth}
    \includegraphics {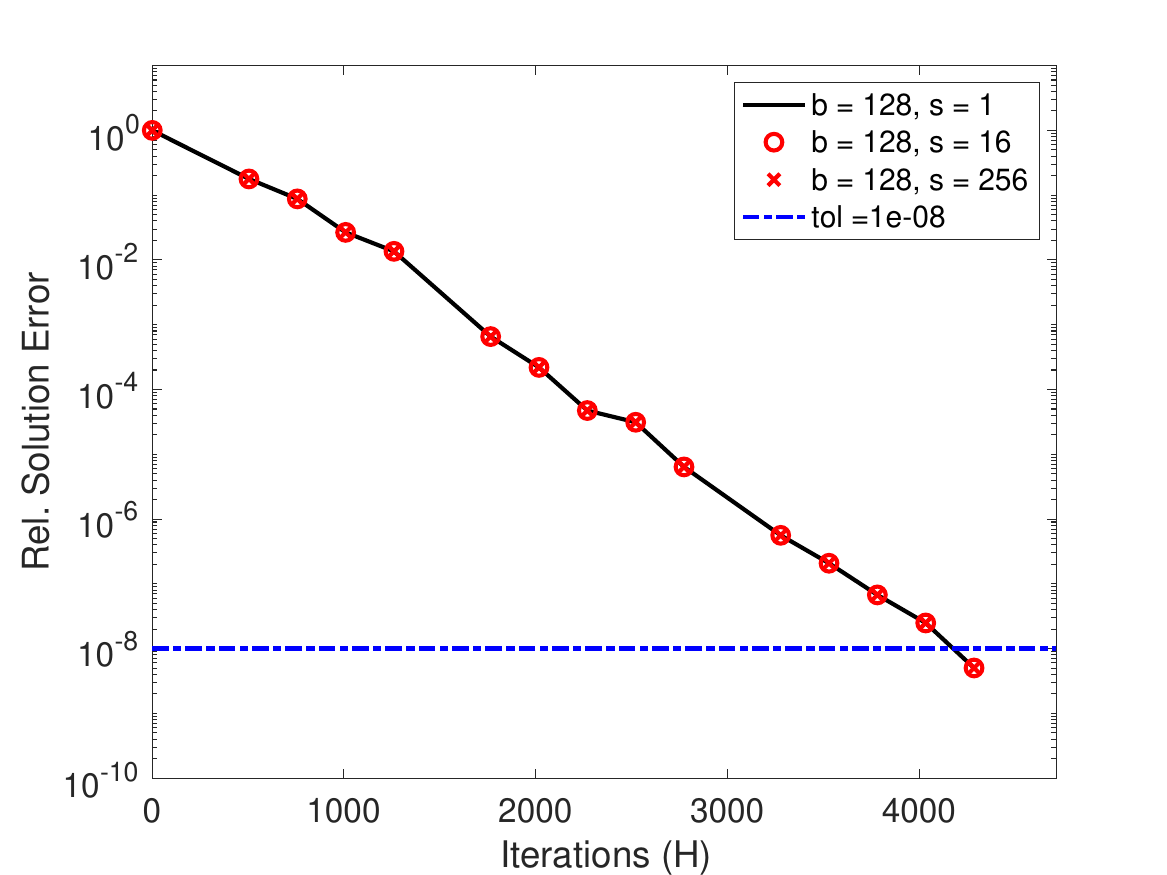}
    \caption{bodyfat, gauss}
  \end{subfigure}\hfill
  \caption{Comparison of BDCD and $s$-step BDCD convergence behavior for K-RR problem.}
  \label{fig:convergence-krr}
\end{figure*}
\Cref{fig:convergence-krr} shows the convergence behavior for the regression datasets in \Cref{table:matlab_datasets} for the three kernels in \Cref{tbl:kernels} for a large setting of $b$ and $s$.
Similar to K-SVM, we observe that $s$-step BDCD attains the same solution at BDCD for every iteration and converges to a relative error tolerance of $10^{-8}$ (single-precision accuracy).
The abalone dataset is the largest MATLAB dataset tested, so we set the block size to $b = 128$.
We report two settings for $s$ to study the convergence behavior at small and large values, where $s = 16$ and $s = 256$, respectively.
We use a smaller block size, $b = 64$, for the smaller bodyfat dataset and use the same settings for $s$.
As \Cref{fig:convergence-krr} illustrates, $s$-step BDCD is numerically stable even when $b \gg 1$.
Furthermore, we can observe that the convergence of $s = 256$ and $s = 16$ (red markers) match for both datasets tested.
Since K-RR, in particular, can utilize $b \gg 1$, we can expect the $s$-step BDCD method to achieve smaller performance gains over classical BDCD due to the bandwidth-latency trade off exhibited by the $s$-step methods.
So we should expect the $s$-step methods to be numerically stable for smaller values of $b$.
\Cref{fig:convergence-ksvm,fig:convergence-krr} both illustrate that the $s$-step methods are numerically stable for the tested ranges and are practical for many machine learning applications.
As a result, the maximum value of $s$ depends on the computation, bandwidth, and latency trade off for a given dataset and hardware parameters of a candidate parallel cluster.
\subsection{Performance Experiments}
In this section, we study the performance of the $s$-step methods for K-SVM and K-RR problems.
We implement all algorithms in C using Intel MKL for sparse and dense BLAS routines and MPI for distributed-memory parallel processing.
We show performance results for each of the kernels shown in \Cref{tbl:kernels}.
The linear and polynomial kernels utilize the Intel MKL SparseBLAS library for sparse GEMM computations.
The polynomial kernel with degree $d$ requires additional elementwise operations on the output of the sparse GEMM.
We compute the RBF kernel by using the definition of the dot-product to expand $\Vert a_{i,:} - a_{j,:}\Vert_2^2$ into a sparse GEMM computation.
Finally, we use the $-O2$ compiler optimization flag additional performance optimization.
The datasets for the performance experiments are shown in \Cref{table:datasets}.
We use datasets obtained from the LIBSVM repository and synthetic sparse dataset in order to study performance characteristics on benchmark machine learning datasets (which may not be load-balanced) and a perfectly load-balanced, sparse synthetic matrix.
All datasets are processed in compressed sparse row (CSR) format.
We use a Cray EX cluster to perform experiments.
We run the DCD and BDCD methods for a fixed number of iterations and vary $s$ for the $s$-step DCD and BDCD methods.
\begin{table}[t]
    \centering
    \begin{tabular}{|c|c|c|c|c|}
        \hline
        Dataset & $m$ & $n$ & nnz & Sparsity \\
        \hline
        colon-cancer& 62 & 2,000 & 124000 & 0 \\
        \hline
        duke breast-cancer & 44 & 7,129 & 313676 & 0 \\
        \hline
        synthetic & 2,000 & 800,000 & 16,000,000 & 99\% \\
        \hline
        news20.binary & 19,996 & 1,355,191 & 909,7916 & 99.97\% \\
        \hline
    \end{tabular}
    \caption{Datasets used in the performance experiments.}
    \label{table:datasets}
\end{table}
\begin{figure*}
    \centering
    \setkeys{Gin}{width=1\linewidth}
    \begin{subfigure}[t]{0.32\textwidth}
      \includegraphics{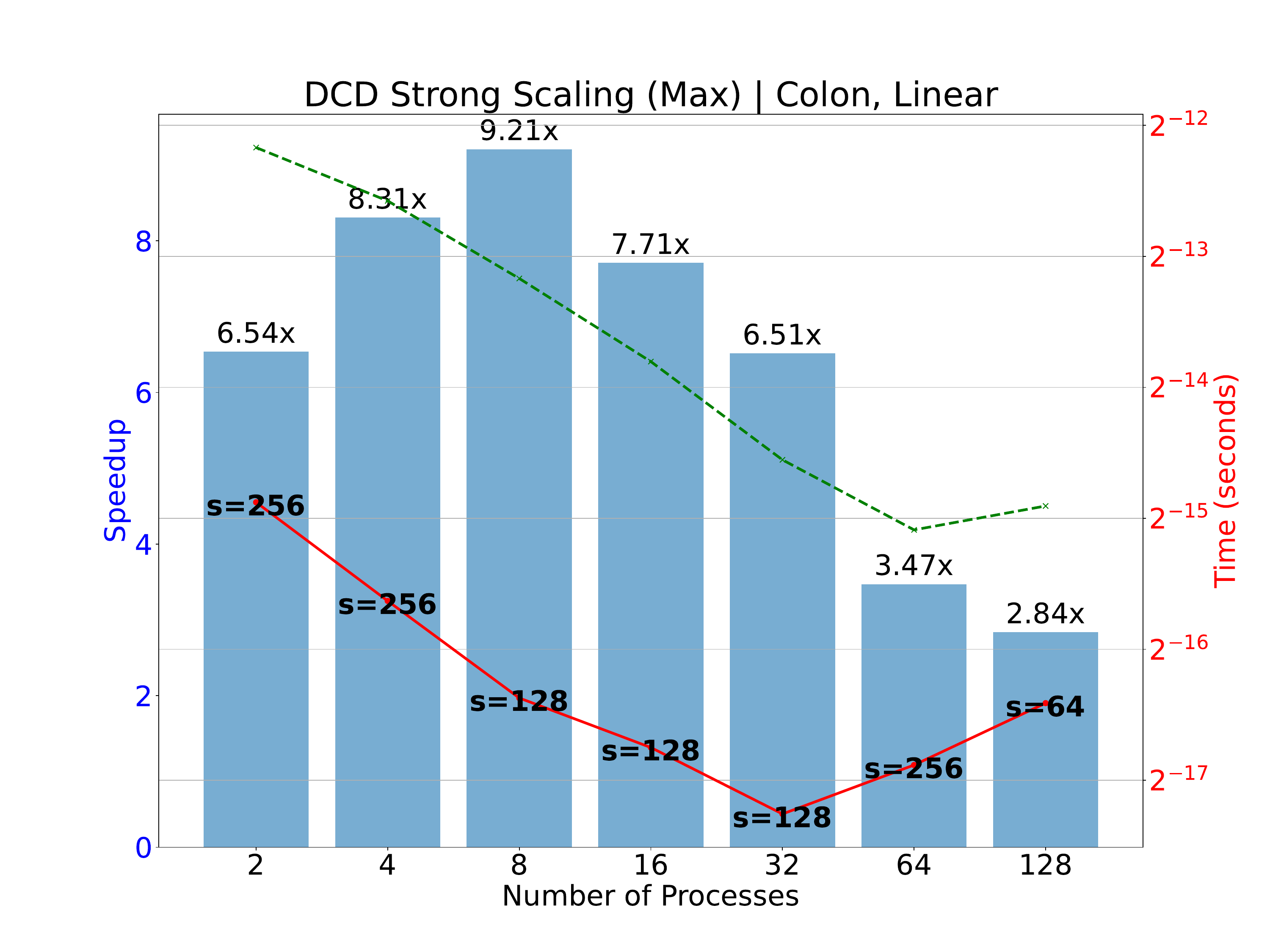}
      \caption{colon cancer, linear}
    \end{subfigure}\hfill
    \begin{subfigure}[t]{0.32\textwidth}
      \includegraphics{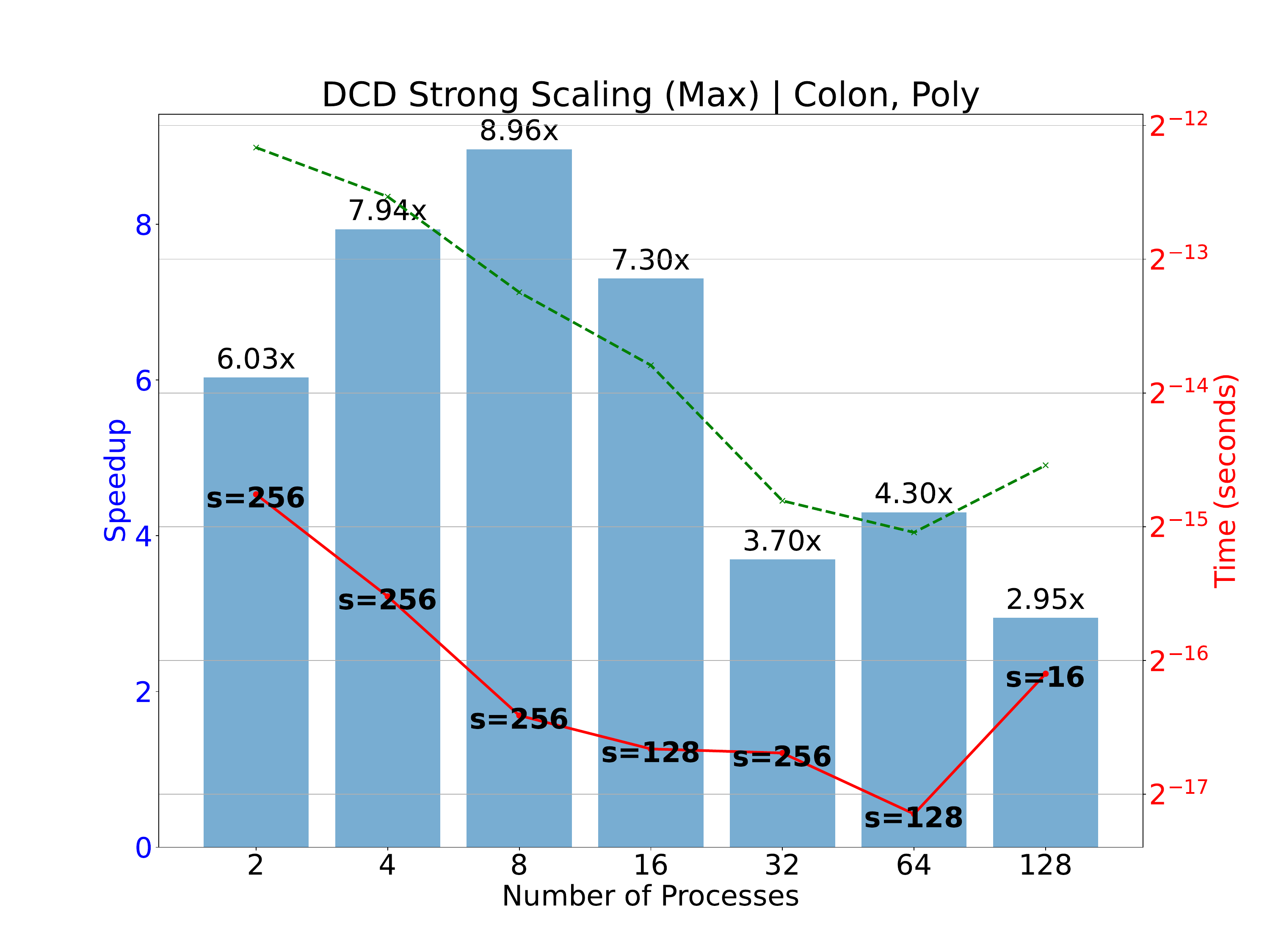}
      \caption{colon cancer, polynomial}
    \end{subfigure}\hfill
    \begin{subfigure}[t]{0.32\textwidth}
      \includegraphics{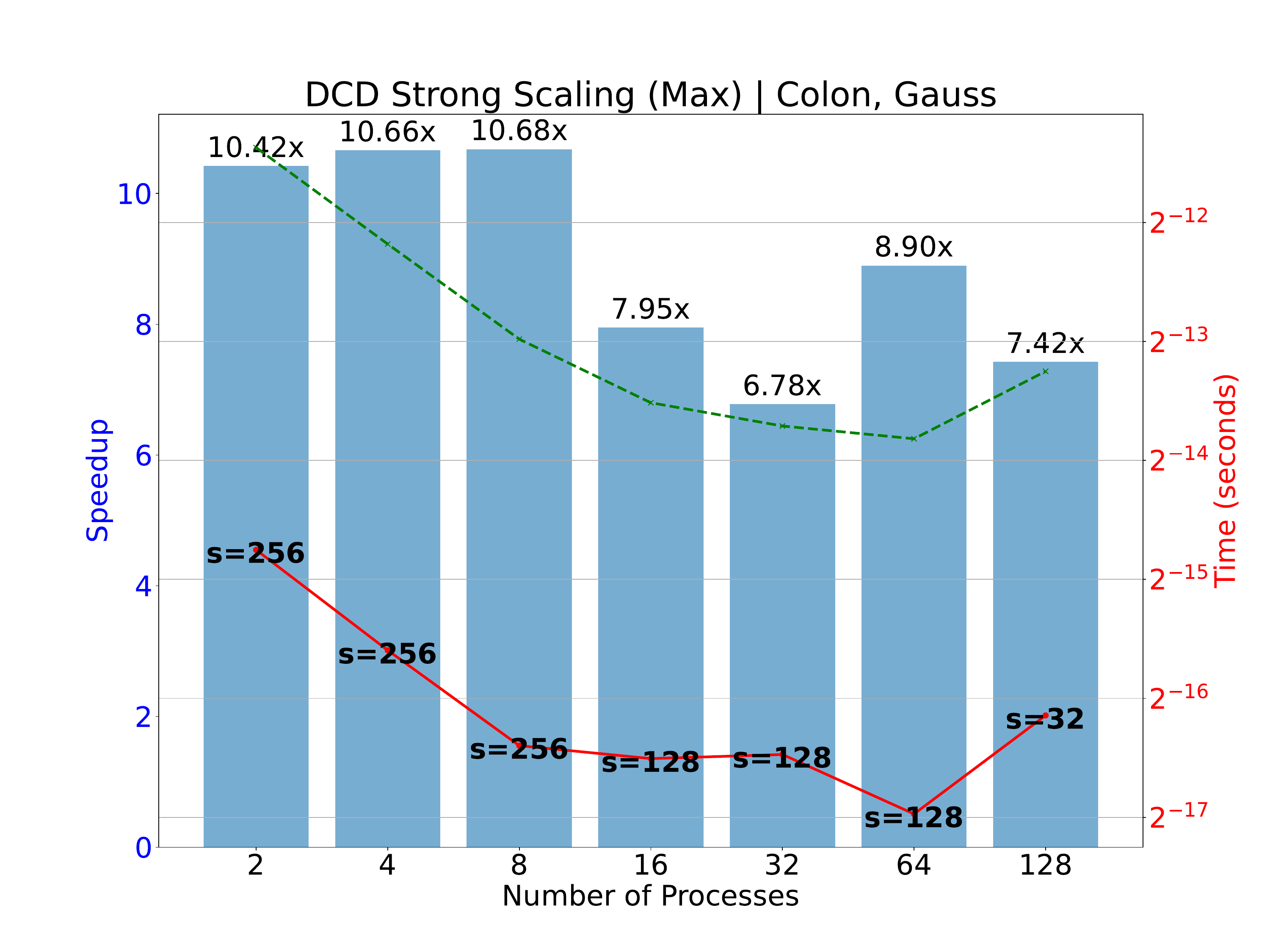}
      \caption{colon cancer, RBF}
    \end{subfigure}\hfill
    \begin{subfigure}[t]{0.32\textwidth}
      \includegraphics{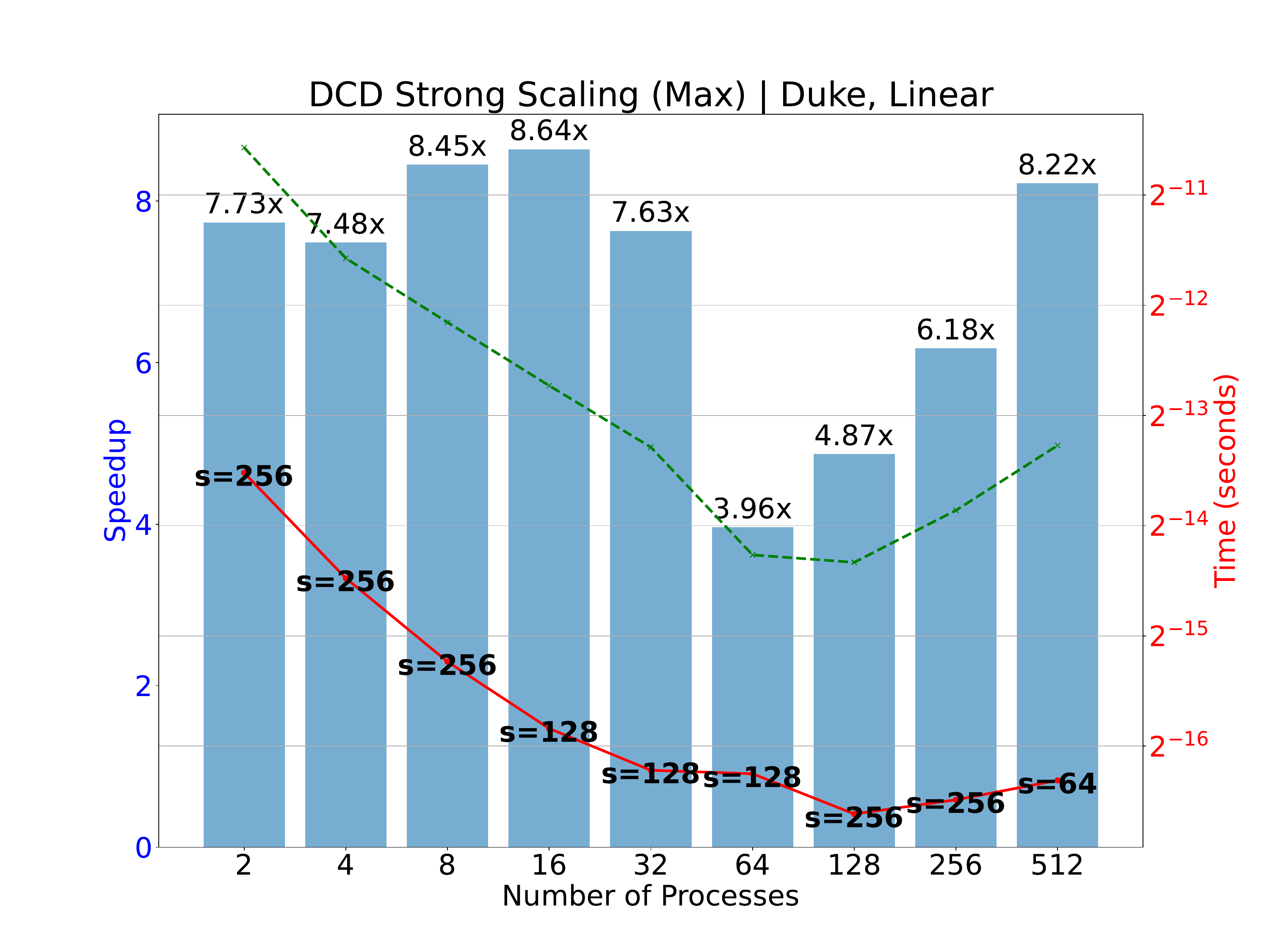}
      \caption{duke, linear}
    \end{subfigure}\hfill
    \begin{subfigure}[t]{0.32\textwidth}
      \includegraphics{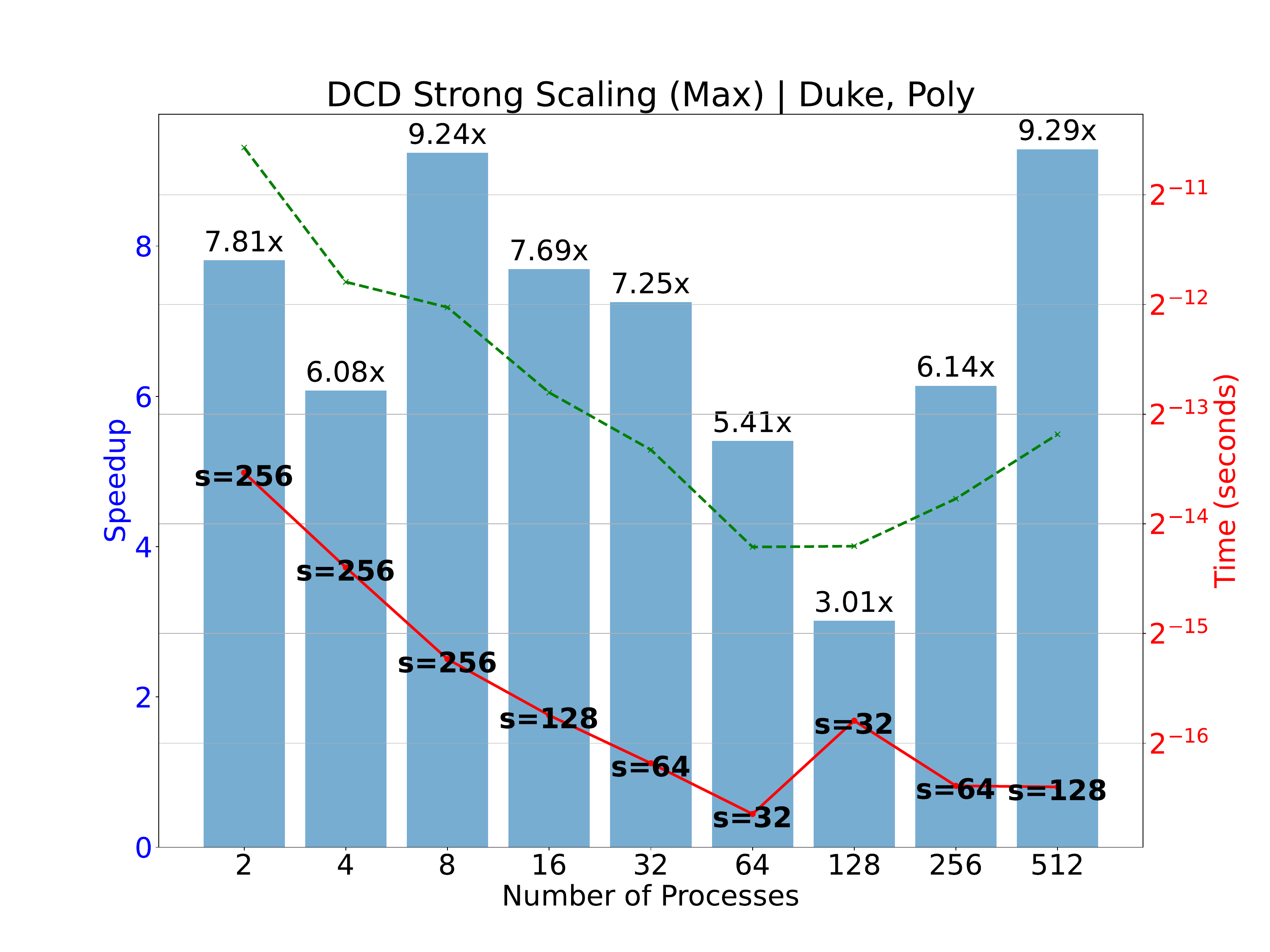}
      \caption{duke, polynomial}
    \end{subfigure}\hfill
    \begin{subfigure}[t]{0.32\textwidth}
      \includegraphics{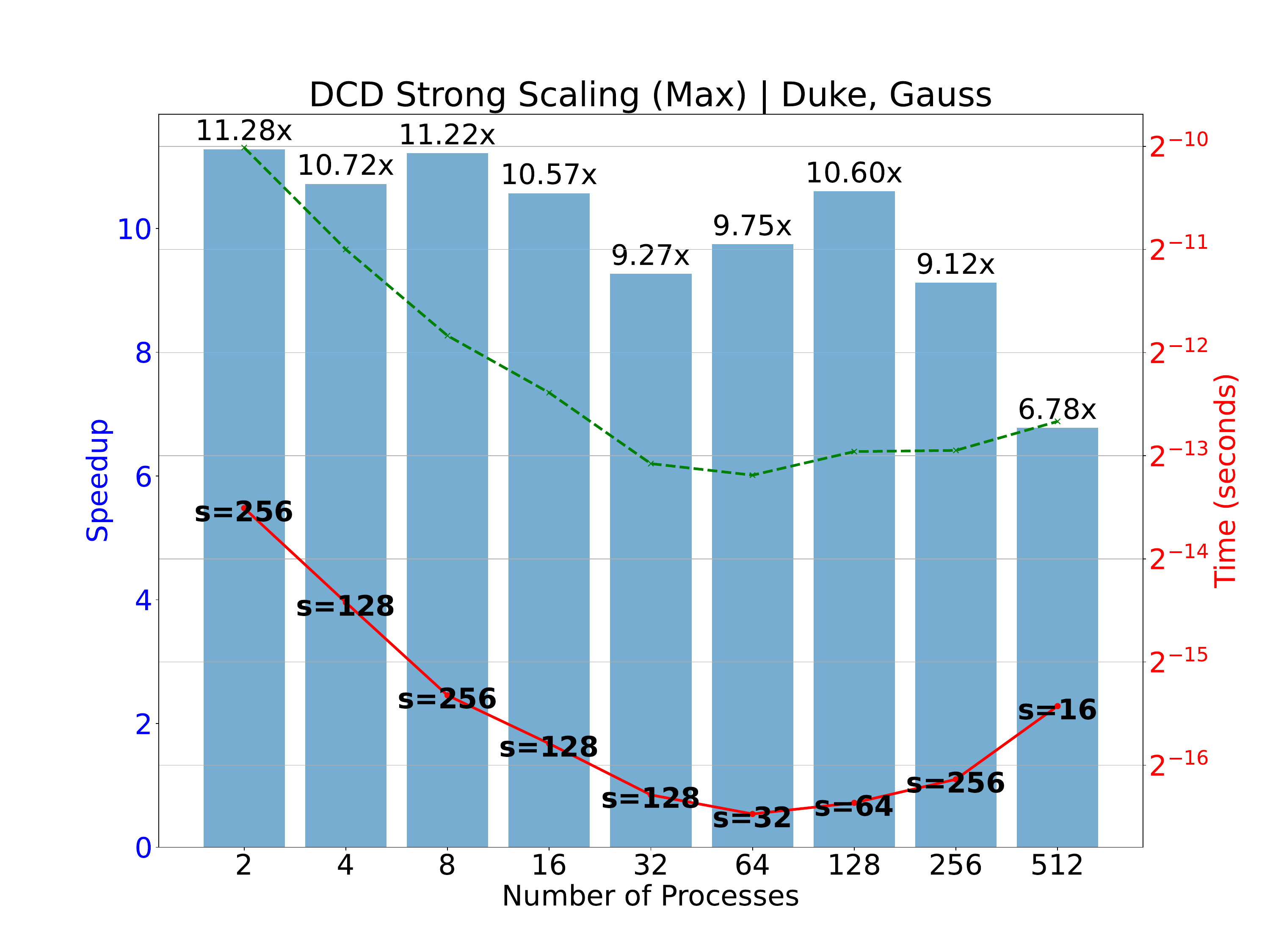}
      \caption{duke, RBF}
    \end{subfigure}
    \begin{subfigure}[t]{0.32\textwidth}
      \includegraphics{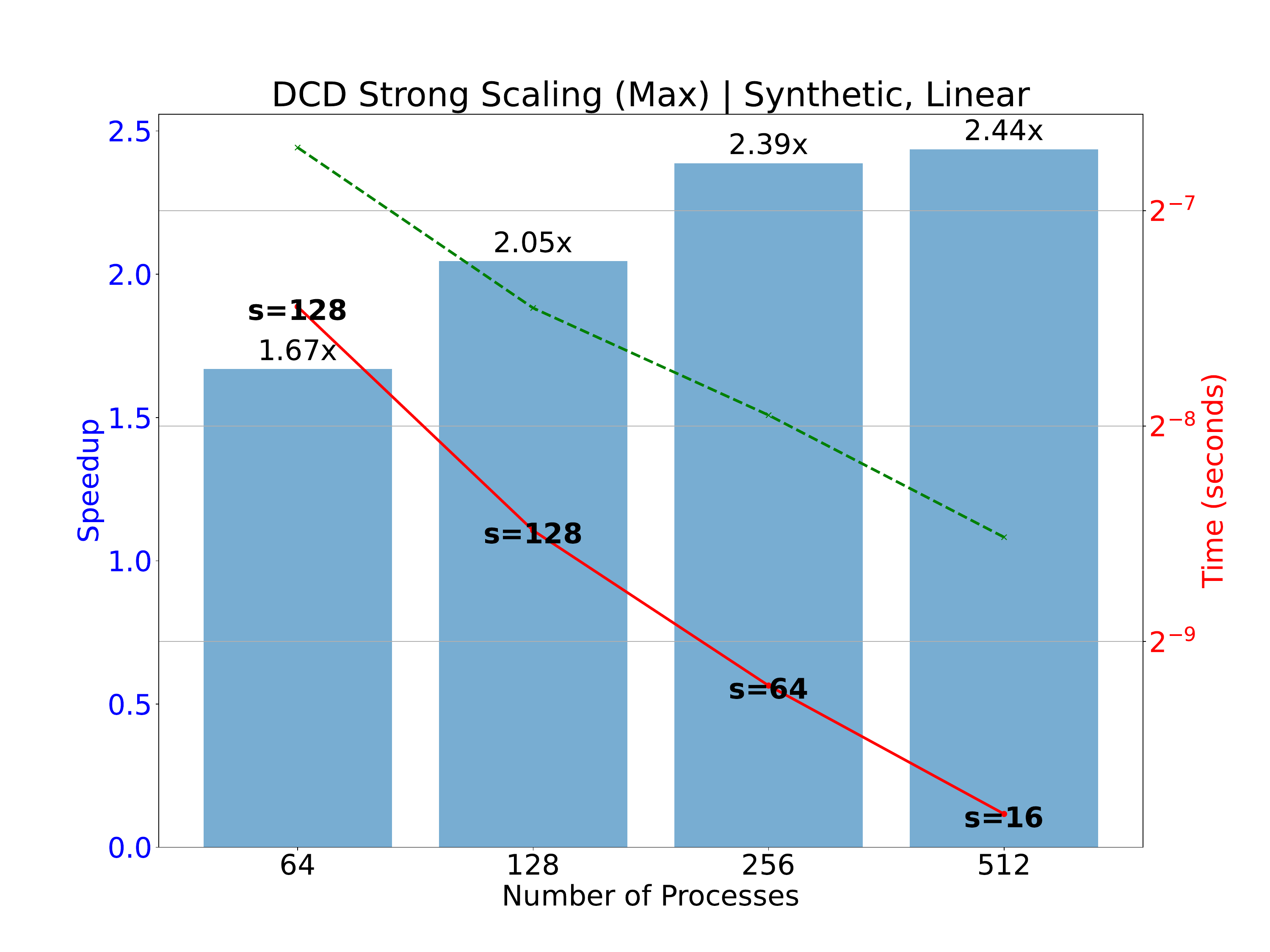}
      \caption{synthtic, linear}
    \end{subfigure}\hfill
    \begin{subfigure}[t]{0.32\textwidth}
      \includegraphics{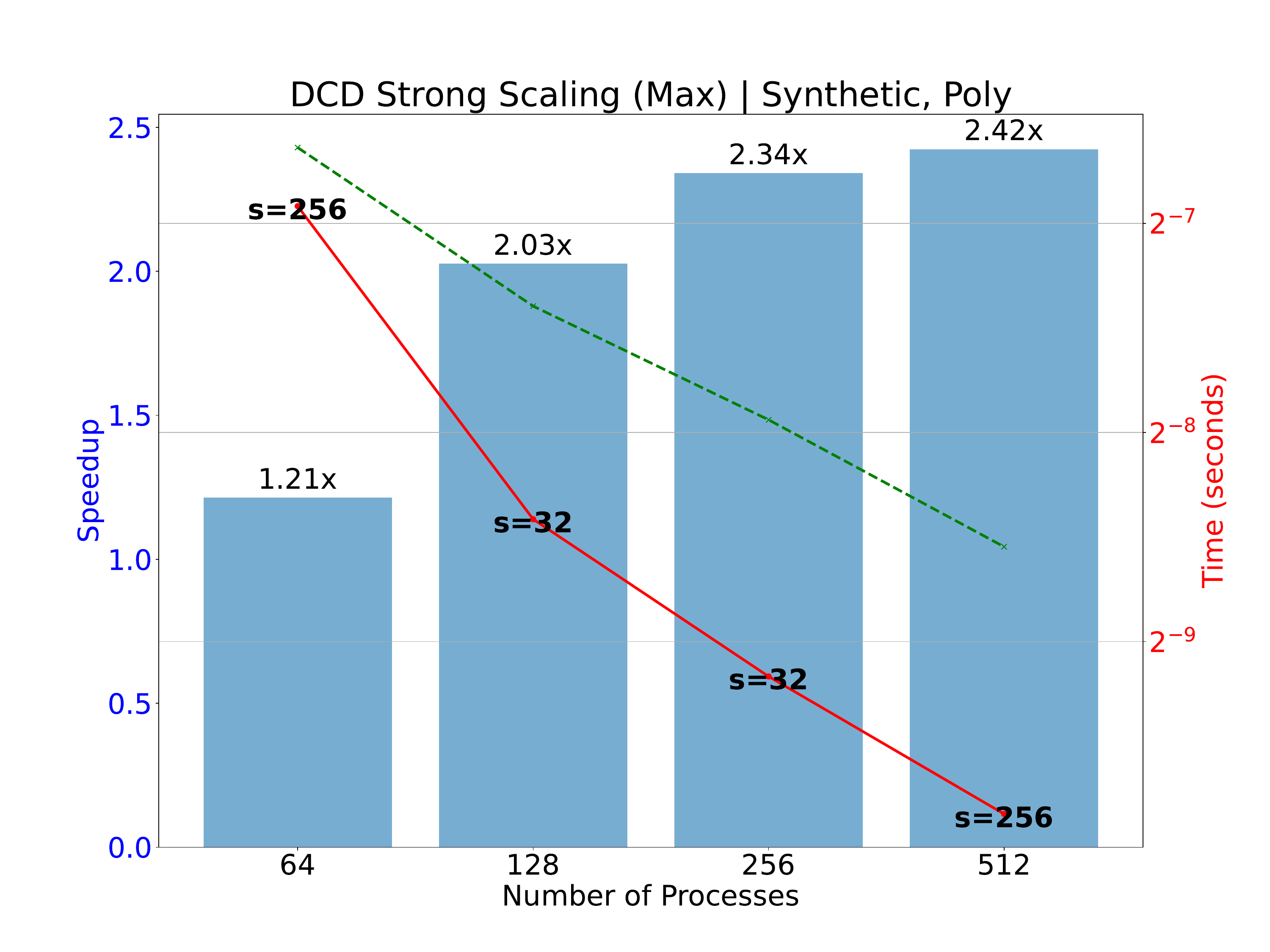}
      \caption{synthetic, polynomial}
    \end{subfigure}\hfill
    \begin{subfigure}[t]{0.32\textwidth}
      \includegraphics{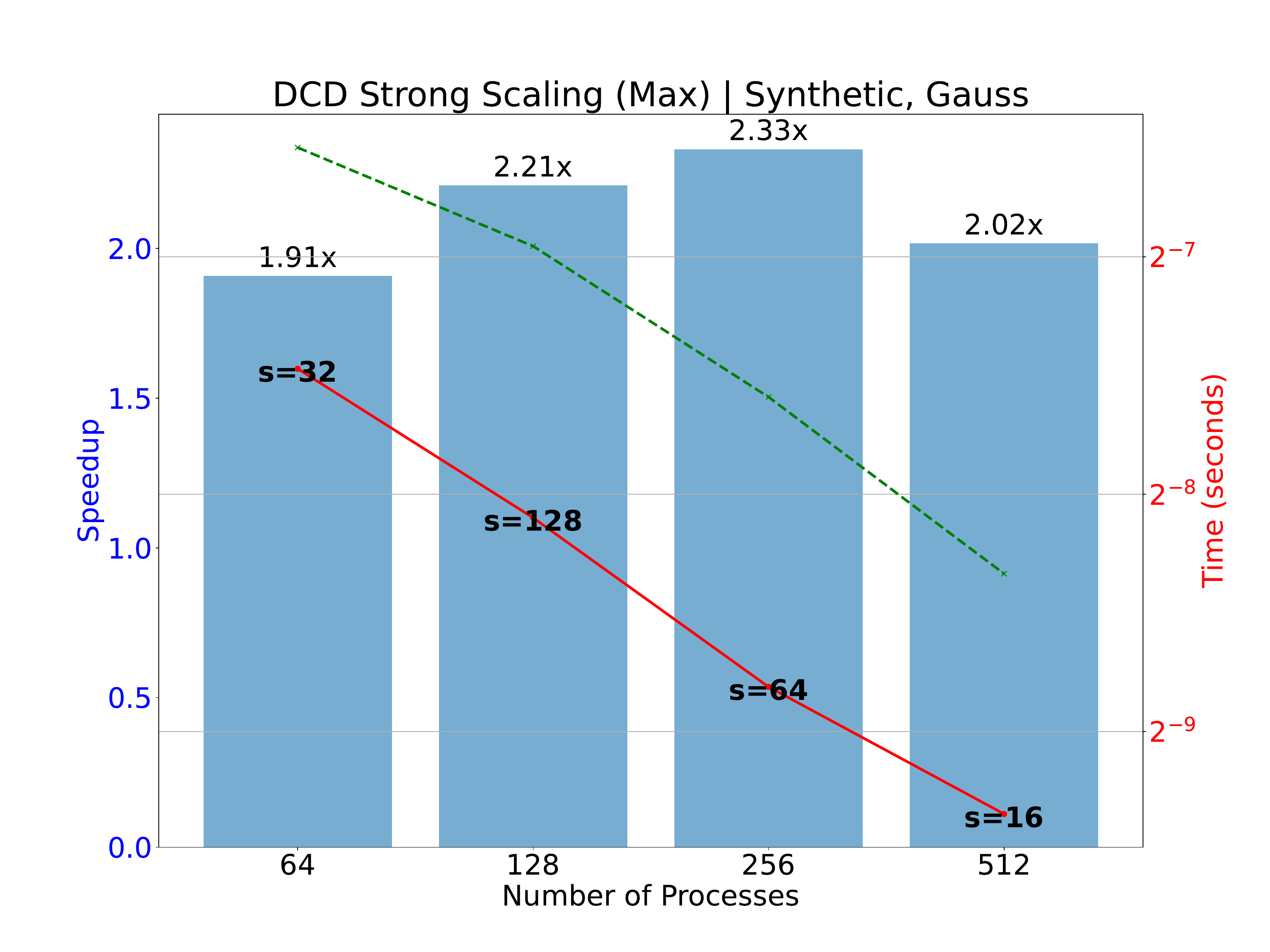}
      \caption{synthetic, RBF}
    \end{subfigure}\hfill
    \caption{Strong Scaling of DCD and $s$-step DCD for K-SVM.}
    \label{fig:strongscaling}
\end{figure*}
We partition the dataset and store it in 1D-column layout (i.e. feature partitioning) so that each MPI process stored roughly $n/P$ columns.
Each CPU node contains two sockets equipped with AMD EPYC 7763 processors with each containing $64$ physical cores.
We did not see any benefits from utilizing simultaneous multi-threading, so we limit the number of MPI processes per node to $128$ processes.
We use MPI process binding and disable dynamic frequency scaling to ensure comparable performance as the number of MPI processes and number of nodes is varied.
\subsubsection{Strong Scaling}\label{subsec:strongscaling}
We explore the strong scaling behavior of the DCD and $s$-step DCD methods for K-SVM in \Cref{fig:strongscaling}.
We also present performance results for BDCD and $s$-step BDCD methods for K-RR, specifically, as the block size is varied.
DCD is limited to $b = 1$ since SVM does not have closed-form solution, thus, we expect better scaling and performance.
Given the computation, bandwidth, and latency trade off, we expect $s$-step BDCD to achieve reduced performance benefits as $b$ is increased.
We report running time and speedups with respect to the slowest processor.
We performed offline tuning of $s$ to obtain the setting which achieves the best running time.
Note that we limit the values of $s$ to powers of two, thus, additional performance may be attainable with fine-grained tuning of $s$.
\begin{figure*}
    \centering
    \setkeys{Gin}{width=1\linewidth}
    \begin{subfigure}[t]{0.32\textwidth}
      \includegraphics{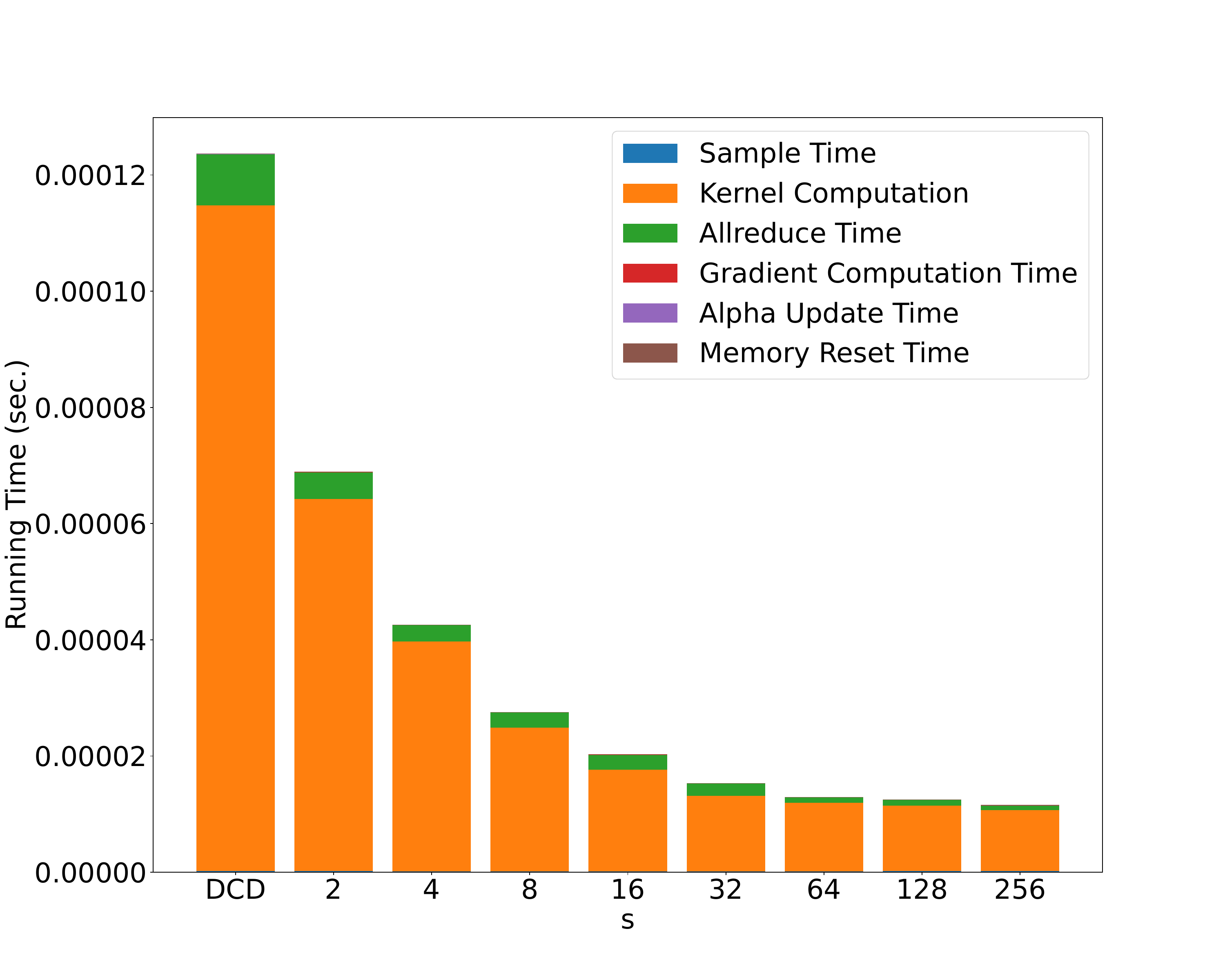}
      \caption{colon-cancer, P = 64, RBF}
    \end{subfigure}\hfill
    \begin{subfigure}[t]{0.32\textwidth}
      \includegraphics{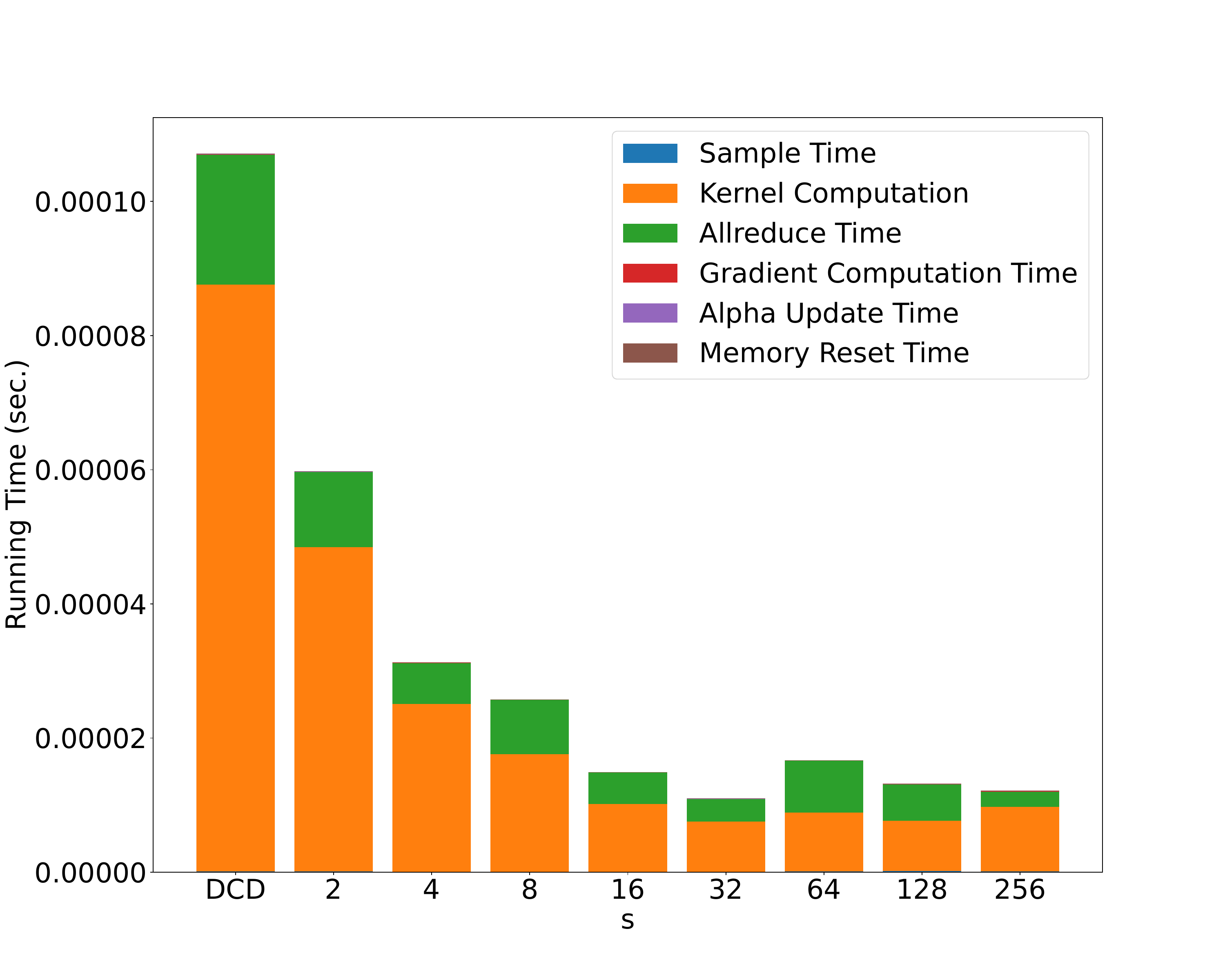}
      \caption{duke, P = 64, RBF}
    \end{subfigure}
    \begin{subfigure}[t]{0.32\textwidth}
      \includegraphics{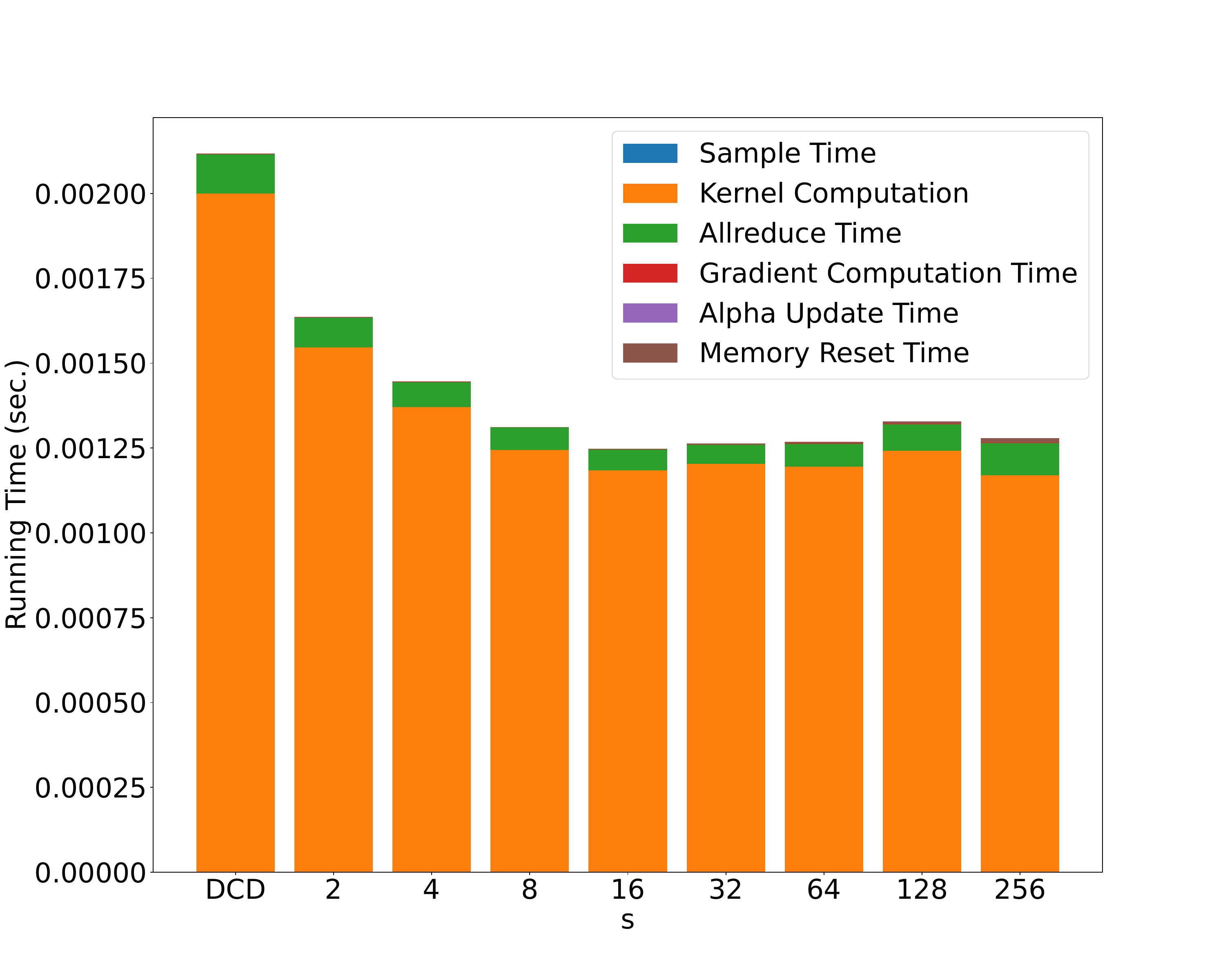}
      \caption{synthetic, P = 512, RBF}
    \end{subfigure}
    \caption{Running Time Breakdown of DCD and $s$-step DCD for values of $P$ with fastest running times.}
    \label{fig:runtime-breakdown}
\end{figure*}
The small, colon cancer dataset exhibits scalability to $O(10)$ processors in \Cref{fig:strongscaling}.
The DCD method is latency bound, so we observe large speedups from the $s$-step DCD method.
On the colon cancer dataset, the $s$-step DCD method attains speedups up to $3.5\times, 4.3\times$, and $8.9\times$ on the linear, polynomial, and RBF kernels, respectively.
On the duke dataset, the $s$-step DCD method attains speedups up to $4.8\times, 5.4\times$, and $9.8\times$ on the linear, polynomial, and RBF kernels, respectively.
Finally, the larger, synthetic dataset attains speedups up to $2.4\times, 2.4\times,$ and $2\times$ on the linear, polynomial, and RBF kernels, respectively.
Finally, the $s$-step BDCD method attains more modest speedups as the block size is increased for all kernels and all datasets, as illustrated in \Cref{table:bdcd-scaling}.
We see decreasing benefits as the block size is increased for the colon-cancer and duke datasets.
For the colon-cancer dataset, we observed speedups of up to $4.78\times$ at $b = 1$ and $1.7\times$ at $b = 4$.
For the colon-cancer dataset, we observed speedups of up to $5.48\times$ at $b = 1$ and $1.68\times$ at $b = 4$.
\begin{table}[h]
    \centering
    \begin{tabular}{|c|c|c|c|c|}
        \hline
        \multirow{2}{*}{Dataset} & \multirow{2}{*}{Kernel} & \multicolumn{3}{c|}{Speedup} \\
        \cline{3-5}
        & & $b = 1$ & $b = 2$ & $b = 4$ \\
        \hline
        \multirow{3}{*}{colon-cancer} & Linear & 4.18 $\times$ & 3.63 $\times$ & 1.86 $\times$ \\
        & Polynomial & 4.08 $\times$ & 2.41 $\times$ & 1.71 $\times$ \\
        & Gauss & 4.78 $\times$ & 3.63 $\times$ & 2.48 $\times$ \\
        \hline
        \multirow{3}{*}{duke} & Linear & 5.48 $\times$ & 3.50 $\times$ & 2.61 $\times$ \\
        & Polynomial & 4.08 $\times$ & 2.41 $\times$ & 1.71 $\times$ \\
        & Gauss & 3.59 $\times$ & 2.54 $\times$ & 1.68 $\times$ \\
        \hline
        \multirow{3}{*}{news20.binary} & Linear & 2.03 $\times$ & 1.32 $\times$ & 1.11 $\times$ \\
        & Polynomial & 1.74 $\times$ & 1.16 $\times$ & 1.09 $\times$ \\
        & Gauss & 1.63  $\times$ & 1.17 $\times$ & 1.11 $\times$ \\
        \hline
    \end{tabular}
    \caption{Speedups attained by $s$-step BDCD over BDCD for solving the K-RR problem for different  block sizes, $b$.}
    \label{table:bdcd-scaling}
\end{table}
\subsubsection{Runtime Breakdown}\label{subsec:runtimebreakdown}
\Cref{fig:runtime-breakdown} presents the running time breakdown of the DCD and $s$-step DCD methods for the colon-cancer, duke, and synthetic datasets for various settings of $s$.
We show results for the values of $P$ that achieve the fastest strong scaling running time and show only the results for the RBF kernel.
A notable observation in \Cref{fig:runtime-breakdown} is the decrease in kernel computation time as we increase $s$.
Since DCD is limited to a block size of $1$, the kernel computation is limited to computing a single row of the kernel matrix at each iteration.
In contrast, the $s$-step DCD method computes $s$ rows of the kernel matrix every (outer) iteration.
As a result, the SparseBLAS routines have better single-node memory-bandwidth utilization (in addition to decreasing latency cost).
\Cref{fig:runtime-breakdown} also shows a decrease in MPI allreduce time, which is expected when latency cost dominates.
However, for the synthetic dataset when $s > 16$, we observe that the allreduce time increases.
Note that this increase in communication time is also expected as the bandwidth term begins to dominate.
Thus, the value of $s$ must be tuned carefully to achieve the best performance.
The $s$-step methods require the same total bandwidth as the classical methods, but the per message bandwidth increases by a factor of $s$.
For the colon-cancer and duke datasets, we see significant improvements in kernel computation and allreduce times with $s = 256$ and $s = 32$ being the optimal setting for colon-cancer and duke, respectively.
For the synthetic dataset, we see a smaller factor of improvement in kernel computation and allreduce times.
This suggests that the synthetic dataset can be strong scaled further before the DCD runtime becomes dominated by DRAM/network communication.
\subsubsection{Performance and Load Balance}\label{subsec:loadimbalance}
\begin{figure*}
    \centering
    \includegraphics[width=.8\columnwidth]{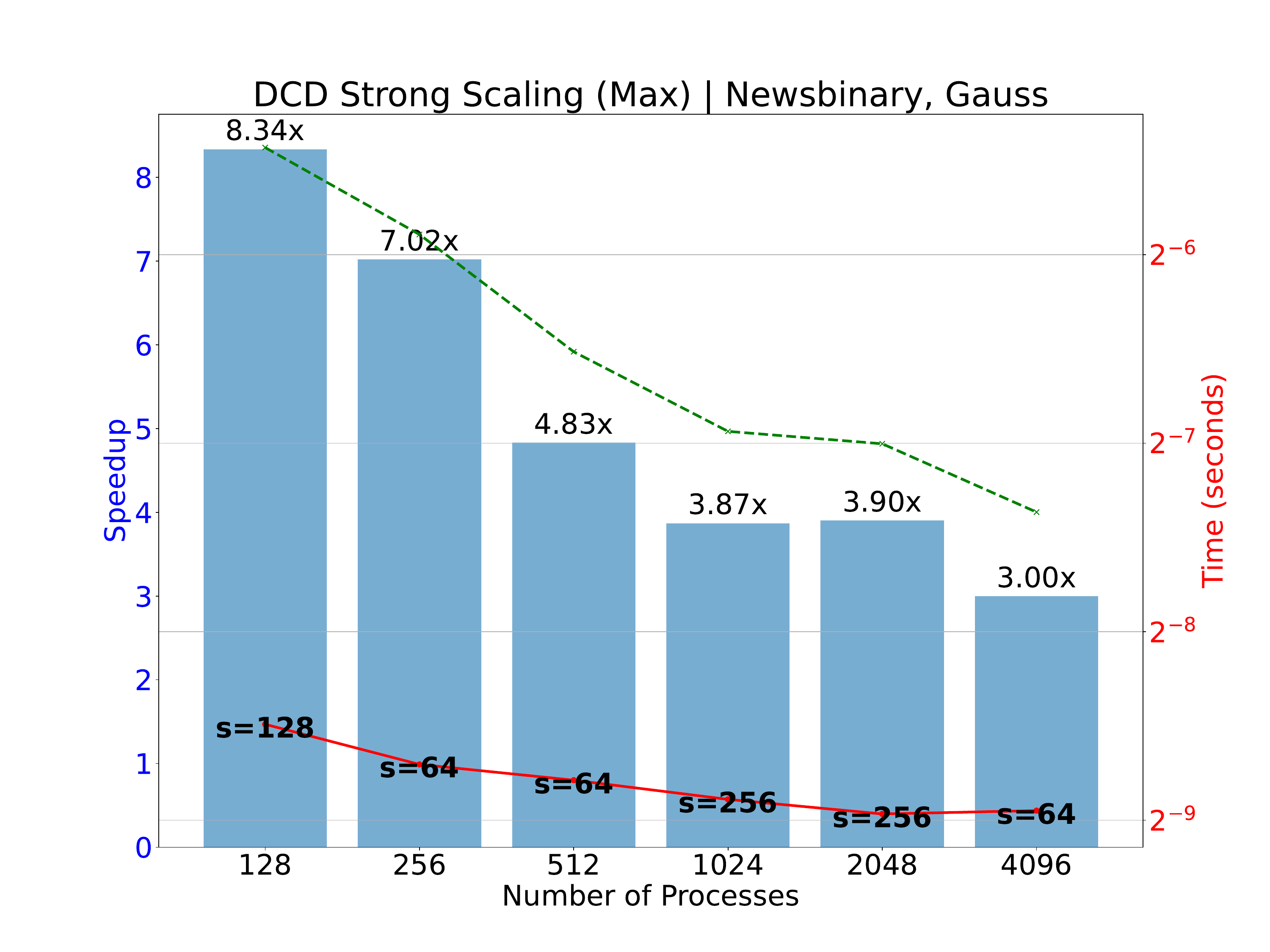}
    \includegraphics[width=.7\columnwidth]{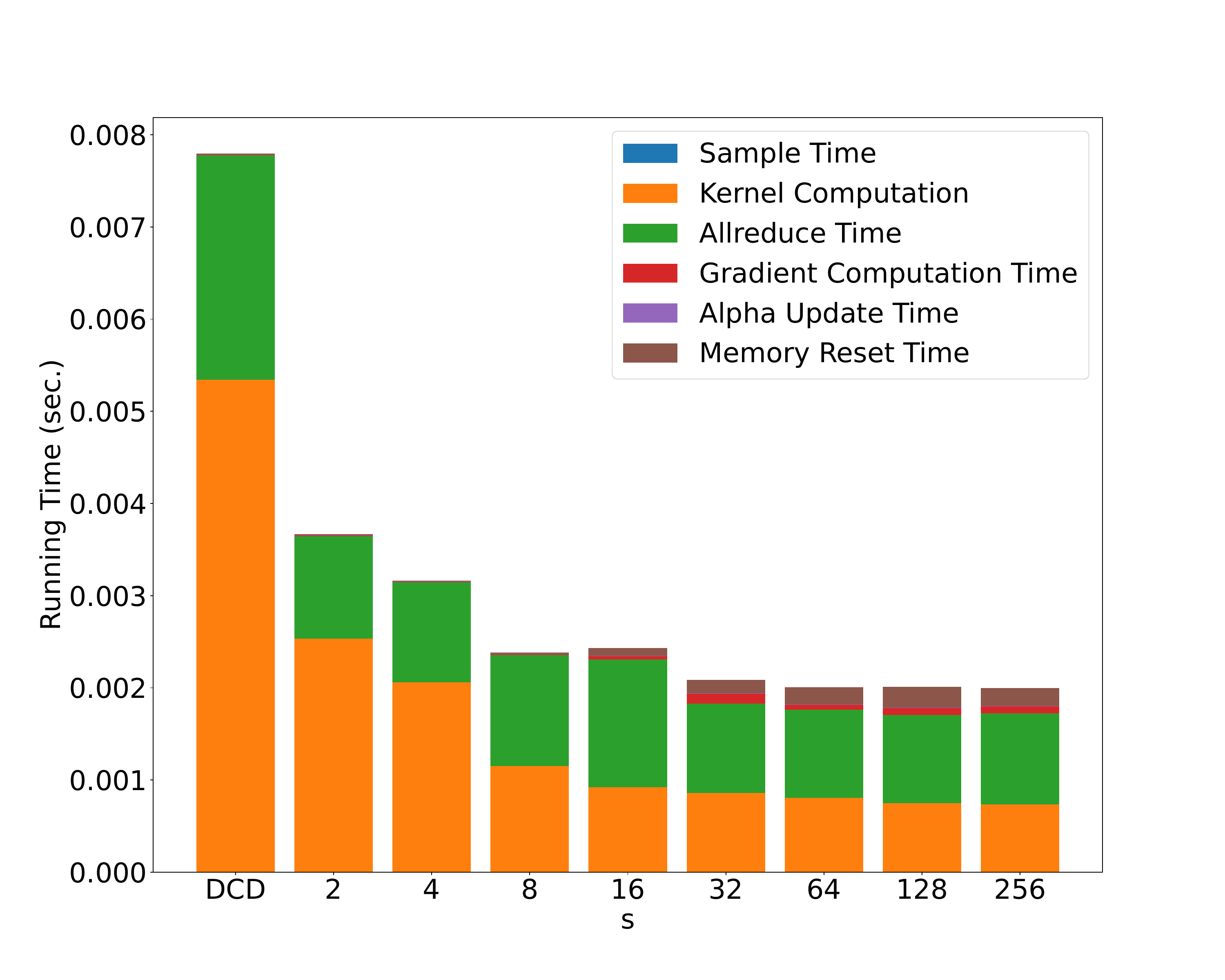}
    \caption{DCD and $s$-step DCD strong scaling and speedup on the news20.binary dataset for K-SVM with RBF kernel.}
    \label{fig:news20-dcd-scaling}
\end{figure*}
Experiments in \Cref{subsec:strongscaling,subsec:runtimebreakdown} showed results for LIBSVM and synthetic datasets which were load balanced.
However, the news20.binary dataset contains non-uniform nonzero distribution which led to load imbalance when stored in 1D-column layout across $P$ processors.
This section explores the performance trade offs of the classical and $s$-step methods under load imbalance specifically for the news20.binary dataset.
\Cref{fig:news20-dcd-scaling} shows the strong scaling behavior of the DCD and $s$-step DCD methods for K-SVM.
DCD under utilizes the available DRAM bandwidth, therefore, the good strong scaling behavior is primarily due to DRAM bandwidth doubling as $P$ is also doubled.
In contrast, $s$-step DCD has more efficient DRAM bandwidth utilization since a block of $s$ rows of the kernel matrix are computed per (outer) iteration.
As a result, $s$-step DCD strong scaling hits the load imbalance scaling limit before DCD.
$s$-step DCD attains a speedup of $3\times$ over DCD at $P = 4096$ with $s = 64$.
\begin{figure*}
    \centering
    \includegraphics[width=.8\columnwidth]{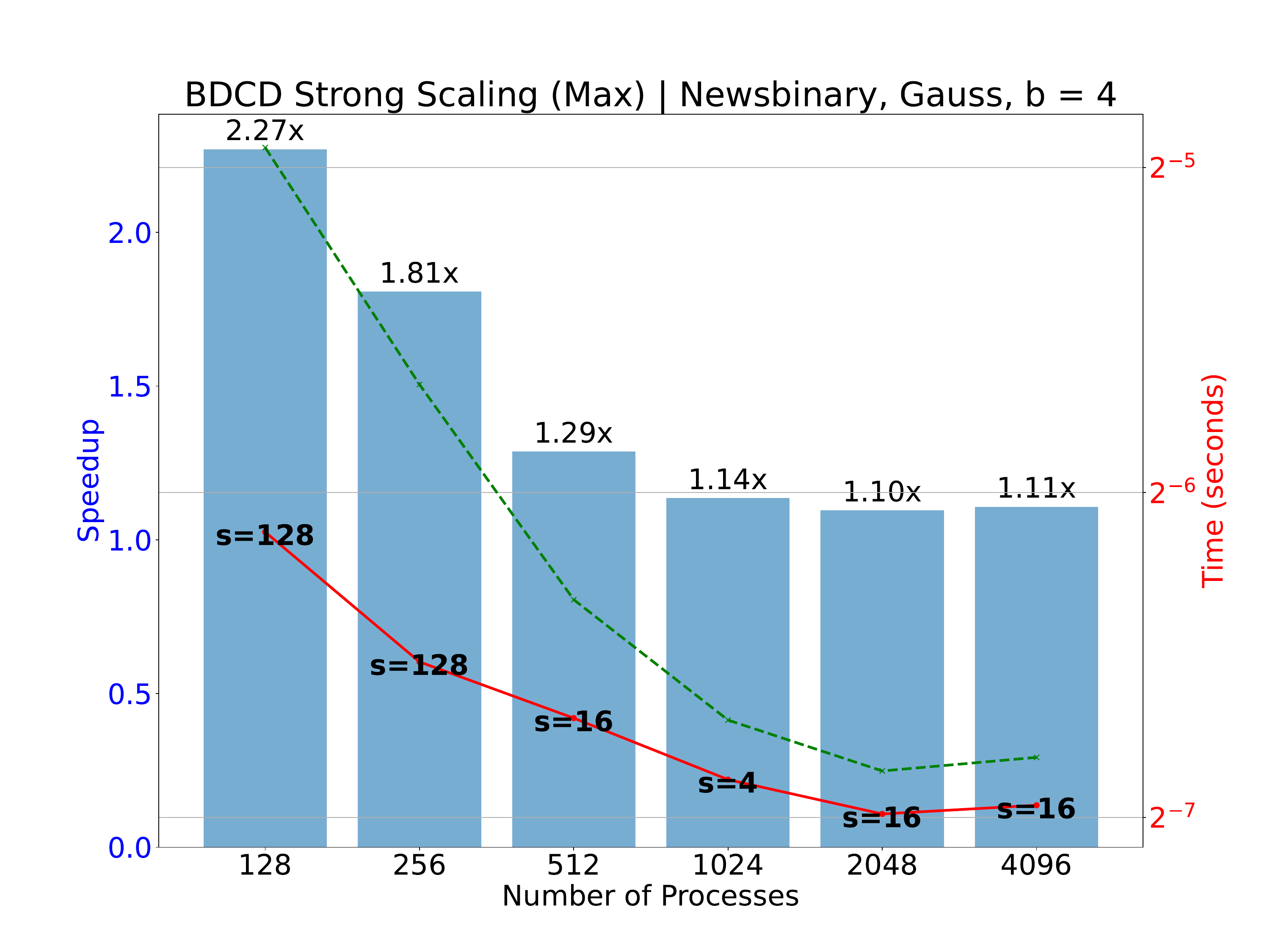}
    \includegraphics[width=.7\columnwidth]{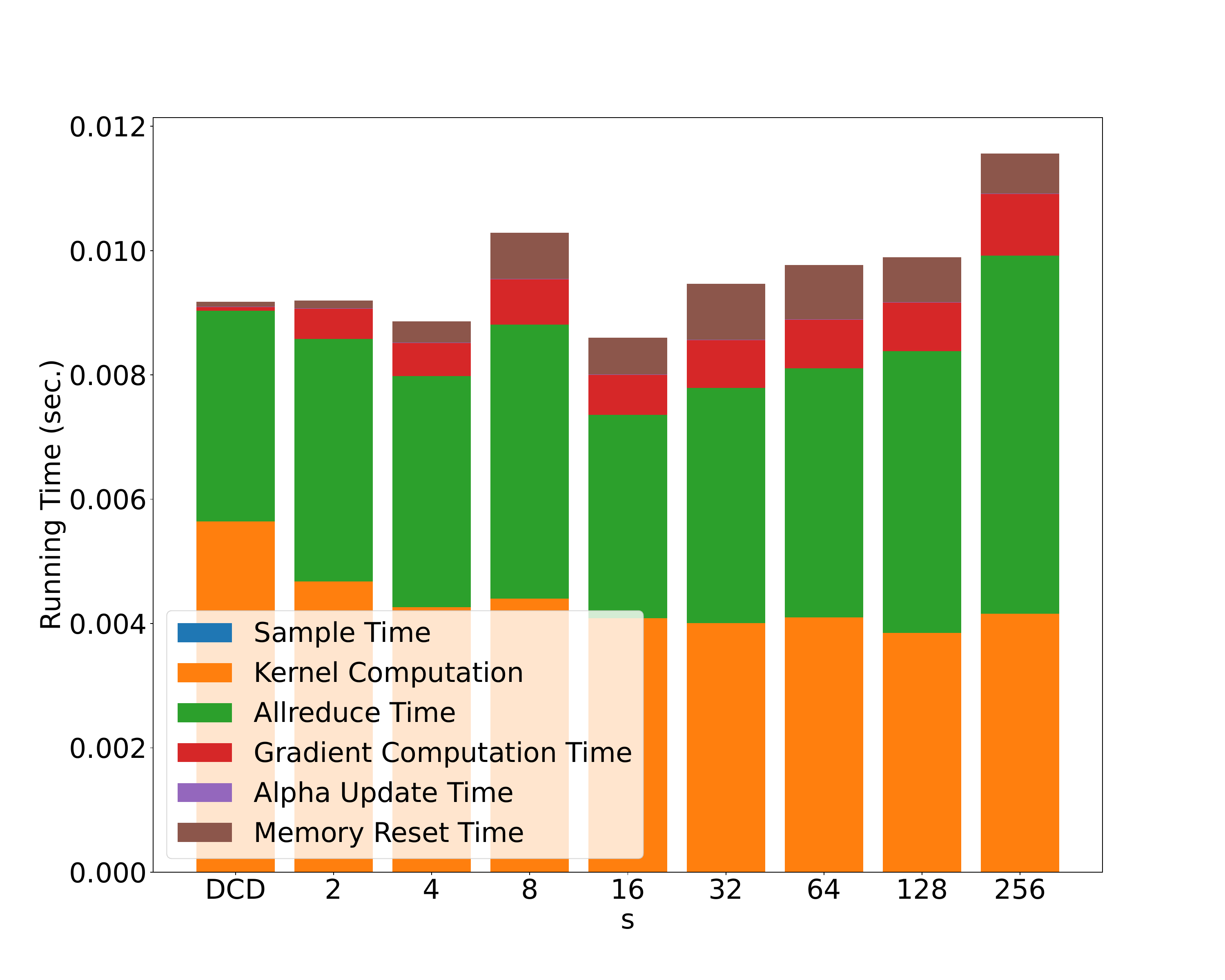}

    \caption{BDCD and $s$-step BDCD strong scaling and speedup on the news20.binary dataset for K-RR with RBF kernel.}
    \label{fig:news20-bdcd-scaling}
\end{figure*}
\Cref{fig:news20-dcd-scaling} shows the running time breakdown of DCD and $s$-step DCD at $P = 2048$ (the fastest $s$-step running time) as $s$ is varied.
The results highlight once again that the $s$-step DCD method reduces both kernel computation and allreduce times for $s > 1$.
Since the news20.binary dataset also has the largest number of samples ($m = \mathcal{O}\left(10^4\right)$), we also see a larger fraction of runtime related to gradient computation and memory management for $s > 1$.
This is due to the factor of $\binom{s}{2} b^2$ additional computation required for the $s$-step method to perform gradient correction.
Once the $s$-step inner loop is computed, the temporary buffers must be reset for the next iteration.
Hence, there is additional running time overhead which increases proportionally with $s$.
However, we note that the gradient correction and memory reset overheads are a small fraction of running time when compared to kernel computation and allreduce times.

\Cref{fig:news20-bdcd-scaling} shows the BDCD and $s$-step BDCD strong scaling and speedup at $b = 4$ with the RBF kernel.
Given the larger block size, both methods exhibit better strong scaling behavior throughout the range of $P$ tested.
The $s$-step BDCD method reaches the load imbalance scaling limit before BDCD, as expected, due to more efficient memory-bandwidth utilization.

\Cref{fig:news20-bdcd-runtime} shows the running time breakdown for BDCD and $s$-step BDCD with $b = 4$ and $P = 2048$ (where BDCD achieves the fastest running time) as $s$ is varied.
Given the larger block size, we observe reduced kernel computation benefits when $s$ is increased when compared to the K-SVM results.
Allreduce time also becomes more bandwidth dominant since $m = \mathcal{O}(10^4)$, therefore, the overall performance benefits of the $s$-step BDCD method reduces to a $1.14\times$ speedup of BDCD.
Furthermore, as $s$ continues to increase we can observe that allreduce bandwidth, gradient correction, and memory reset overhead become a larger fraction of running time.
This suggests that we cannot set both $s$ and $b$ to very large values.
inverse trend is observed with allreduce time; it becomes increasingly significant with higher $s$ values and a greater number of processes, both of which are bandwidth-dominated cases.
For instance, at $2048$ processes, the allreduce time constitutes over $45\%$ of the total runtime at $s = 256$, compared to less than $20\%$ for the same $s$ value at $128$ processes.
Figures 3 (d - f) further illustrate that larger batch sizes tend to exacerbate the dominance of allreduce time. In the case of the news20.binary dataset, there is an increase in overall runtime despite the reduction in kernel computation time as \textit{s} increases. This phenomenon is also evident in the colon-cancer dataset. As depicted in Figure 4 (a - b), the CA-BDCD algorithm continues to reduce the running time until \textit{s} reaches $32$.
Beyond this point, an increase in kernel computation time is observed. The proportion of allreduce time relative to the total runtime also grows with the number of processes; it is significantly less dominant at $p = 4$ of the runtime than at $p = 32$.
In the case of the DCD algorithm, a similar pattern is observed, as both DCD and CA-DCD algorithms involve the computation and communication of $AA^T$.
Specifically, the CA-DCD algorithm efficiently reduces running time with increasing values of $s$.
An increase in the number of processes correlates with a higher proportion of allreduce time.
However, the absolute value of the allreduce time remains relatively stable across different process counts.
This stability is attributed to the fact that the size of the message communicated during allreduce does not depend on the number of processes.
Instead, it is the consistent reduction in kernel computation time that results in the allreduce time occupying a larger proportion of the overall runtime.
\begin{figure*}
  \centering
  \setkeys{Gin}{width=1\linewidth}
  \begin{subfigure}[t]{0.32\textwidth}
    \includegraphics{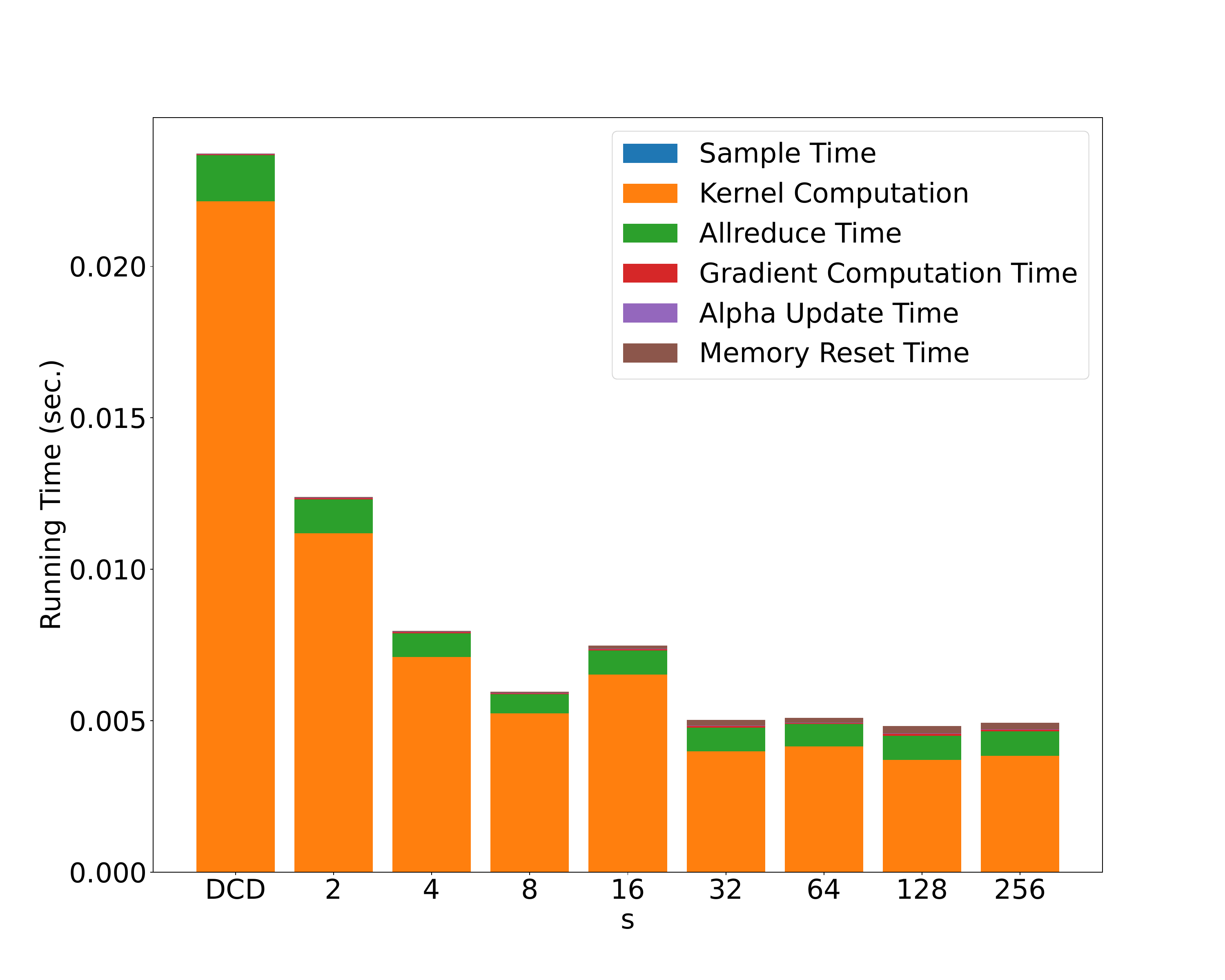}
    \caption{news20.binary, np = 128, gauss}
  \end{subfigure}\hfill
  \begin{subfigure}[t]{0.32\textwidth}
    \includegraphics{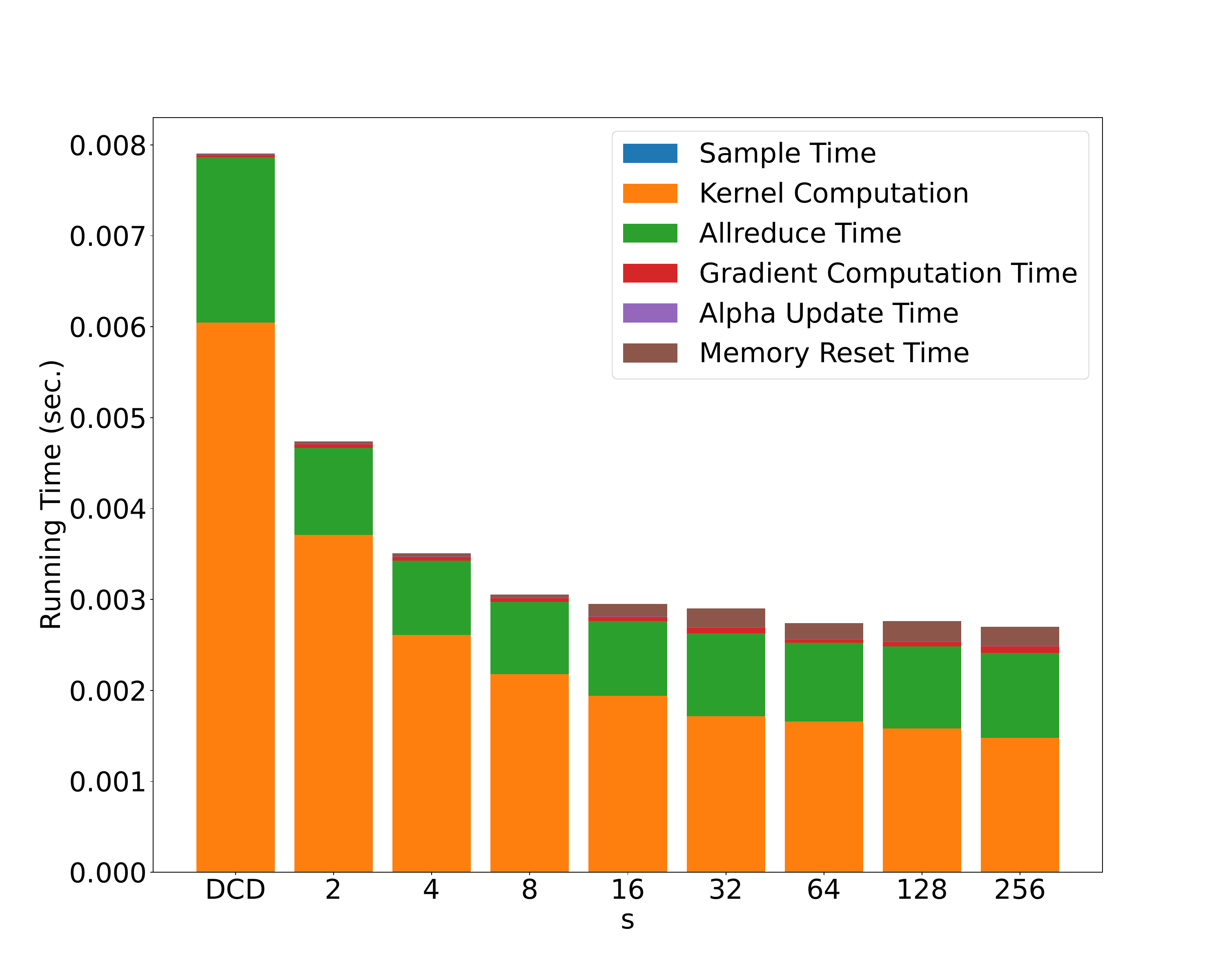}
    \caption{news20.binary, np = 512, gauss}
  \end{subfigure}\hfill
  \begin{subfigure}[t]{0.32\textwidth}
    \includegraphics{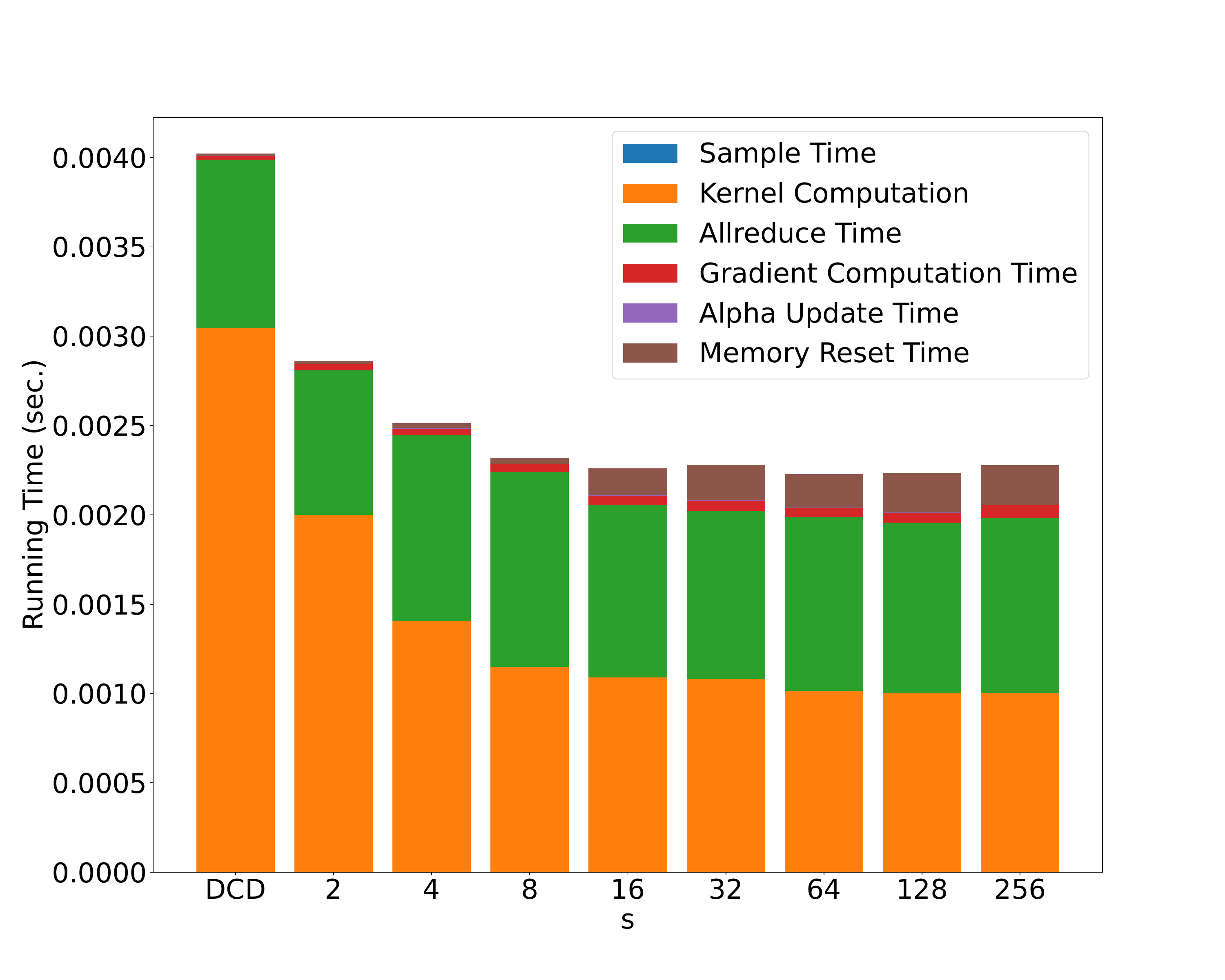}
    \caption{news20.binary, np = 2048, gauss}
  \end{subfigure}\hfill
    \setkeys{Gin}{width=1\linewidth}
  \begin{subfigure}[t]{0.32\textwidth}
    \includegraphics{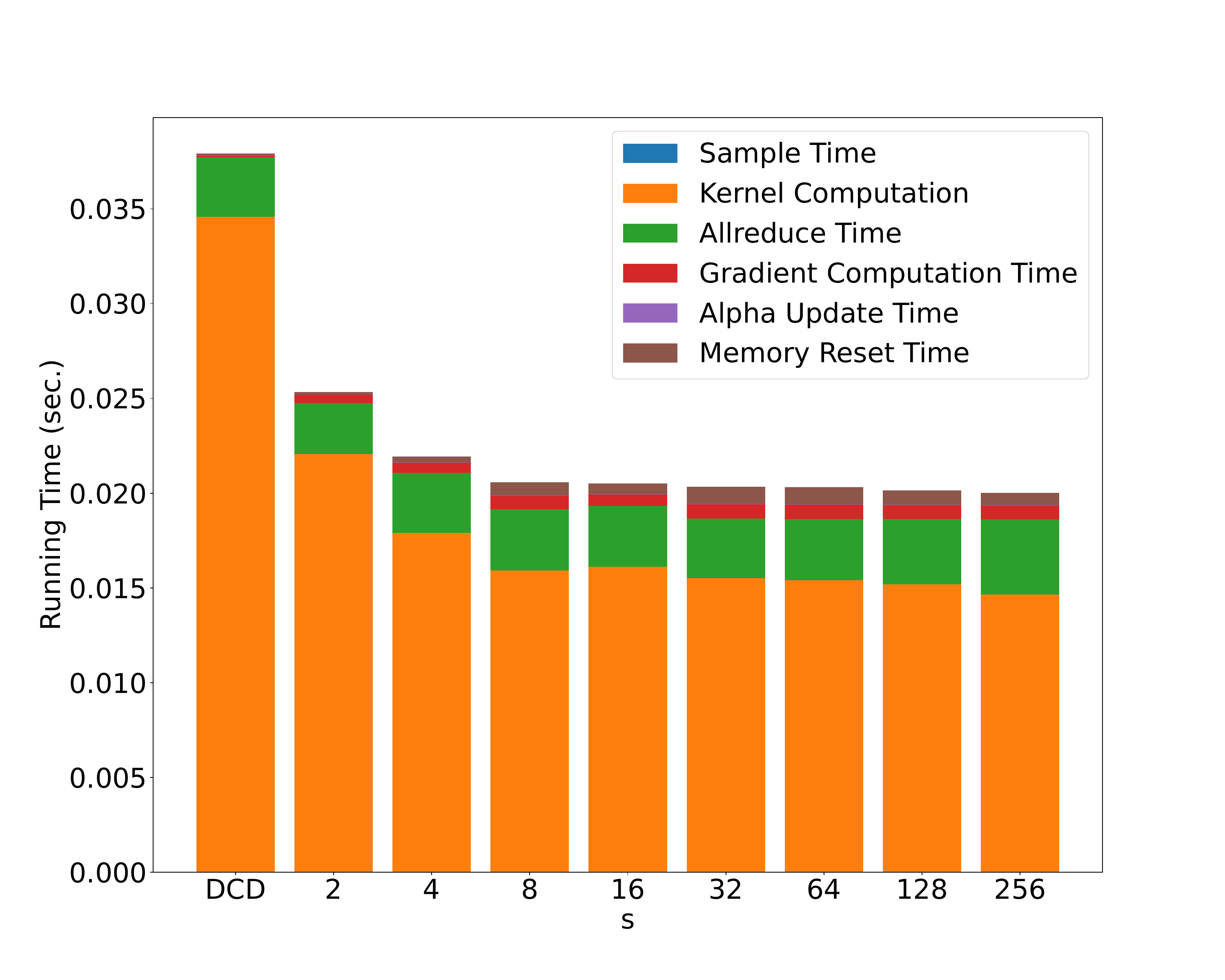}
    \caption{news20.binary, np = 128, batch 4, gauss}
  \end{subfigure}\hfill
  \begin{subfigure}[t]{0.32\textwidth}
    \includegraphics{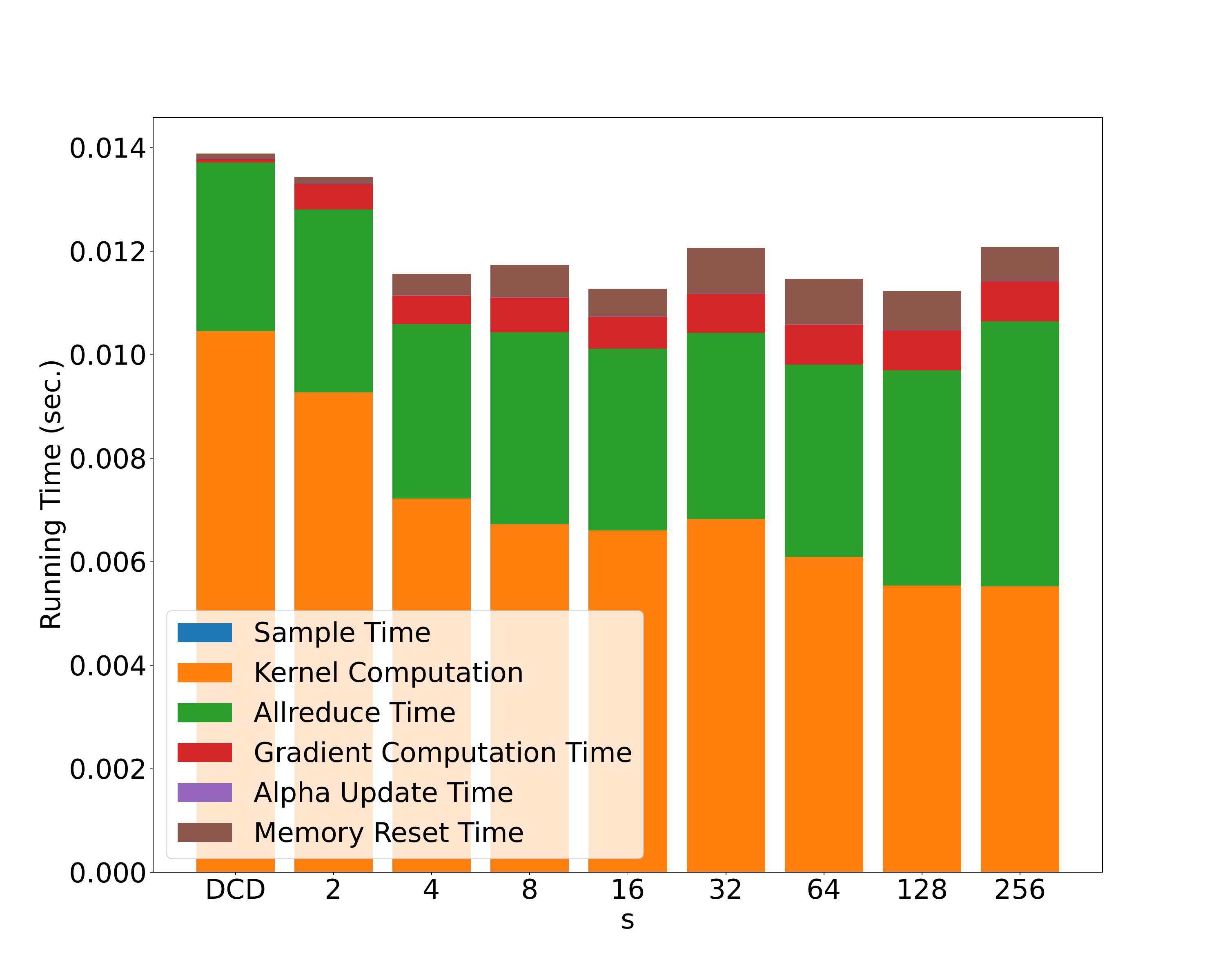}
    \caption{news20.binary, np = 512, batch 4, gauss}
  \end{subfigure}\hfill
  \begin{subfigure}[t]{0.32\textwidth}
    \includegraphics{figures/ridge/composition/Newsbinary_gauss_np2048_b4_Composition.pdf}
    \caption{news20.binary, np = 2048, batch 4, gauss}
  \end{subfigure}\hfill
  \caption{Running Time Breakdown CA-BDCD vs BDCD, Newsbinary}
  \label{fig:news20-bdcd-runtime}
\end{figure*}
\begin{figure*}
  \centering
  \setkeys{Gin}{width=1\linewidth}
  \begin{subfigure}[t]{0.48\textwidth}
    \includegraphics{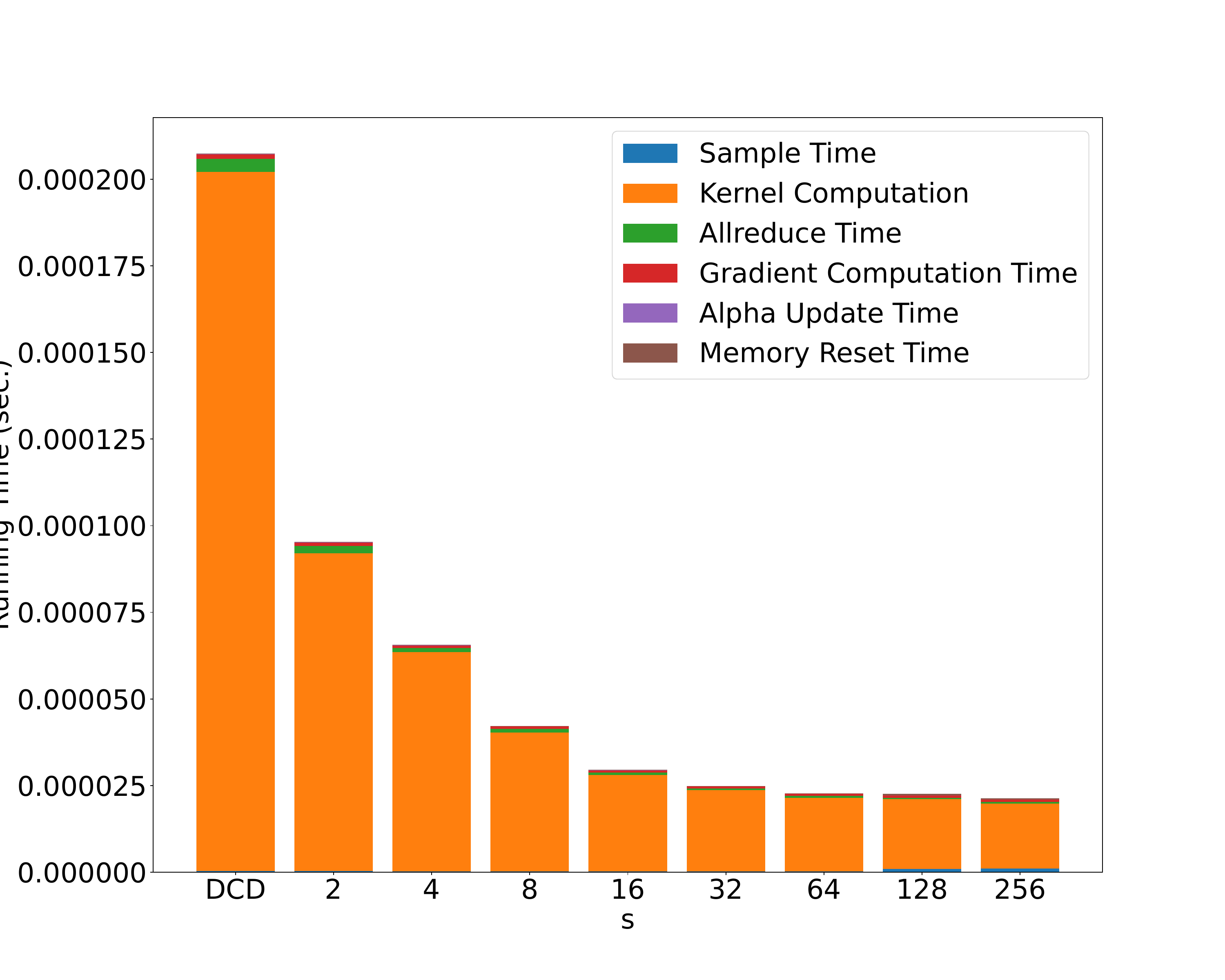}
    \caption{colon-cancer, np = 4, b = 1, gauss}
  \end{subfigure}\hfill
  \begin{subfigure}[t]{0.48\textwidth}
    \includegraphics{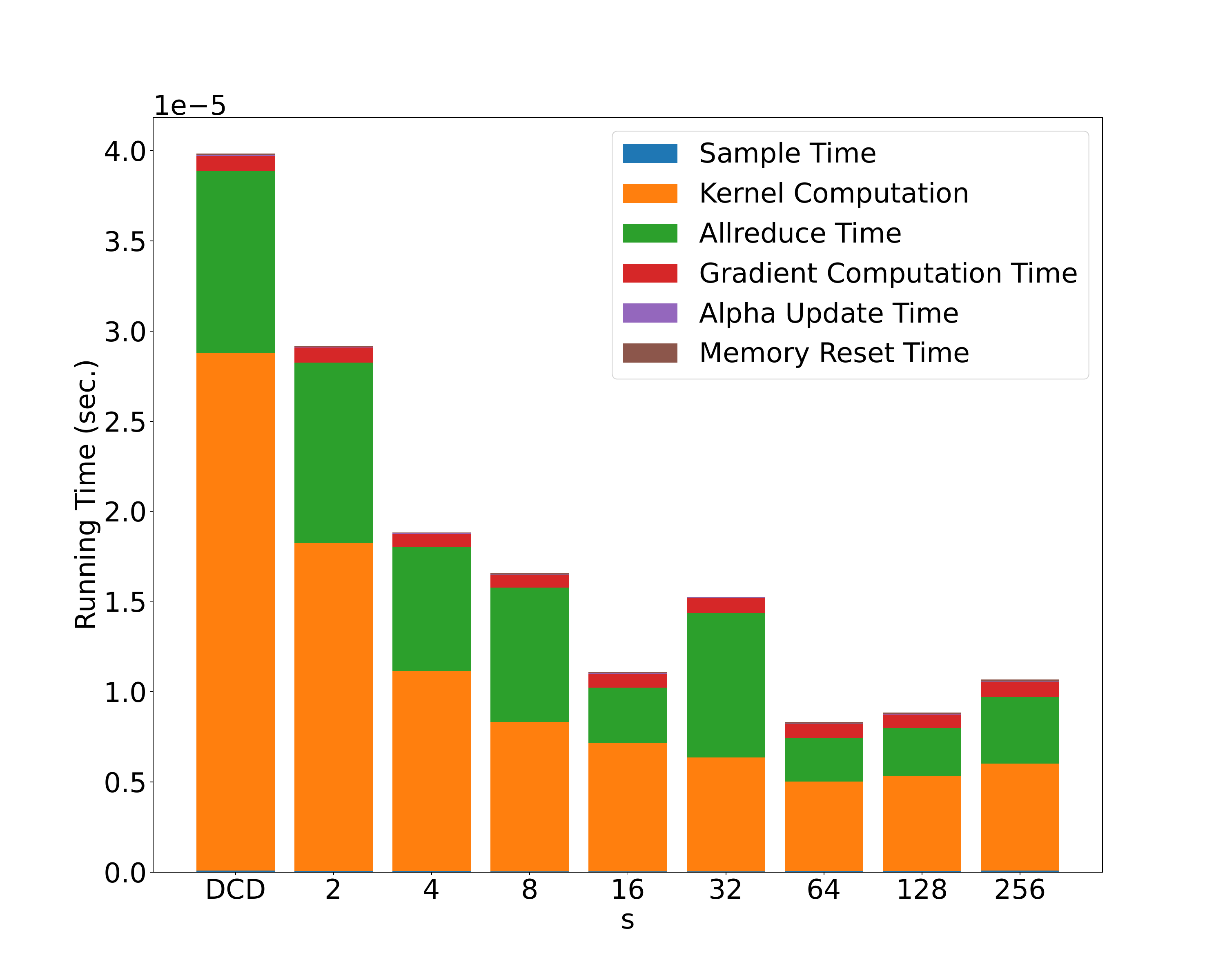}
    \caption{colon-cancer, np = 32, b = 1, gauss}
  \end{subfigure}\hfill

 \setkeys{Gin}{width=1\linewidth}
  \begin{subfigure}[t]{0.48\textwidth}
    \includegraphics{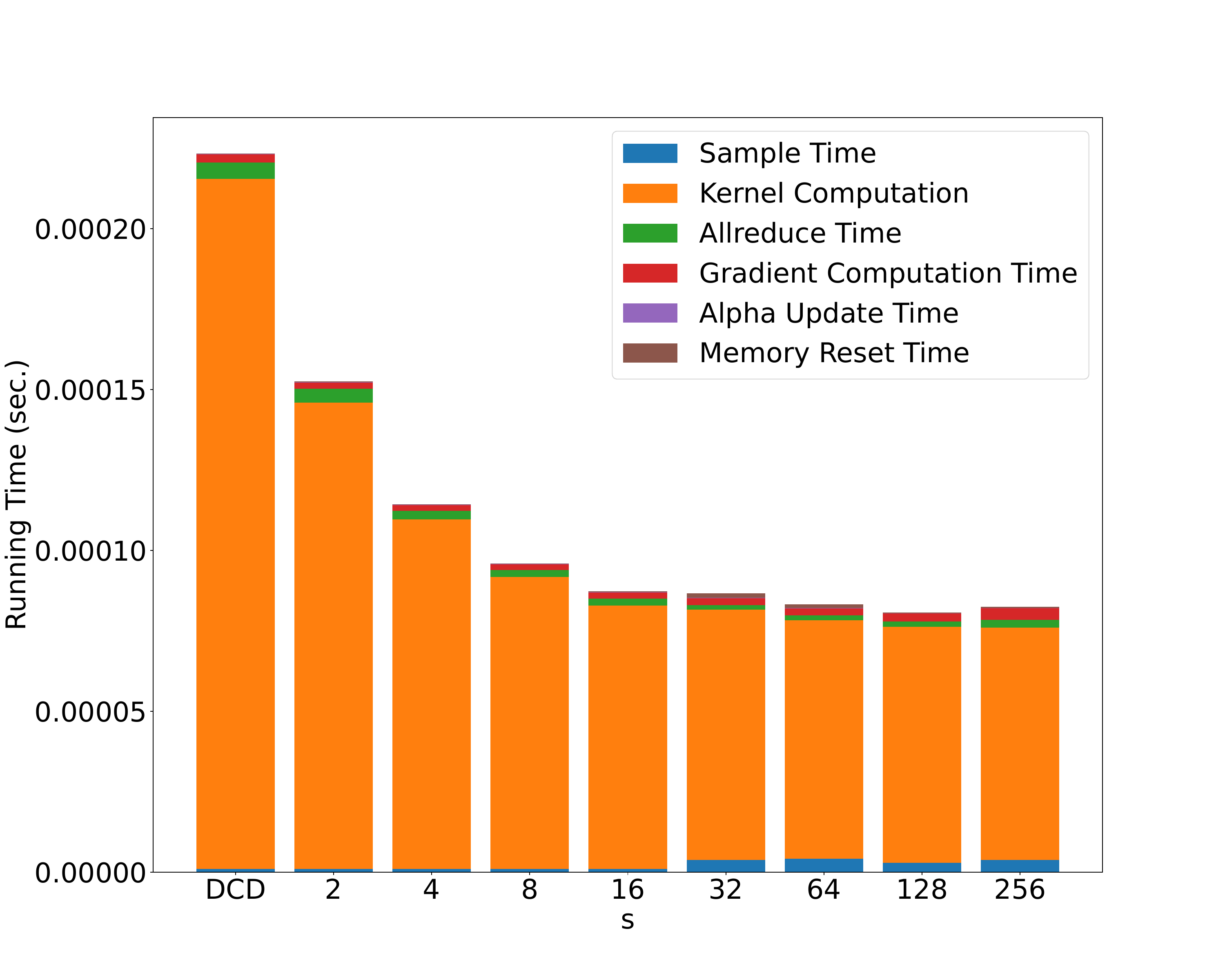}
    \caption{colon-cancer, np = 4, b = 4, gauss}
  \end{subfigure}\hfill
  \begin{subfigure}[t]{0.48\textwidth}
    \includegraphics{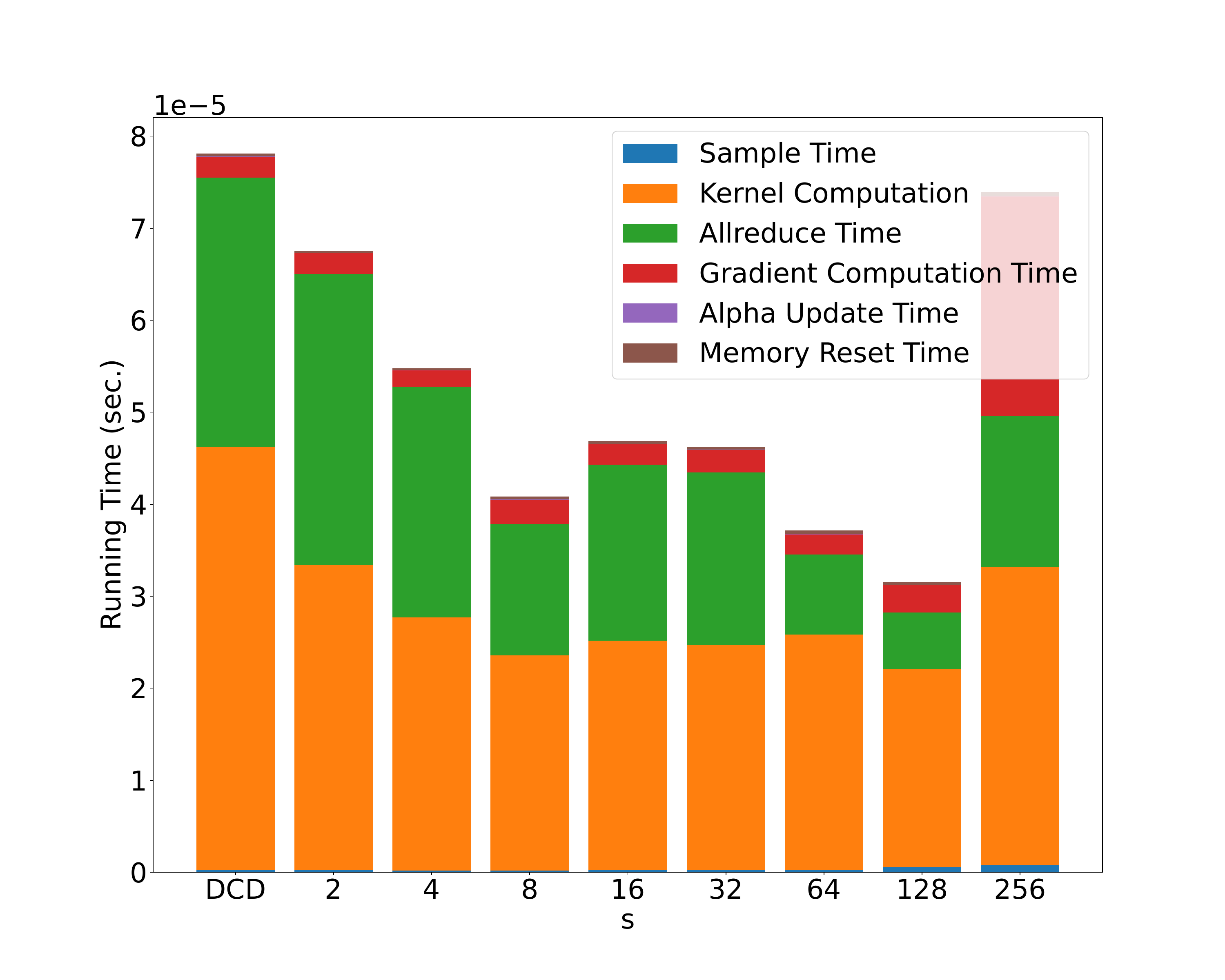}
    \caption{colon-cancer, np = 32, b = 4, gauss}
  \end{subfigure}\hfill
  \caption{Time Composition CA-BDCD vs BDCD, Colon}
\end{figure*}

\section{Conclusion}
This work demonstrates that the $s$-step DCD and $s$-step BDCD methods for K-SVM and K-RR, respectively, attain large speedups when latency is the dominant cost.
We show that this conclusion holds for dense and sparse datasets as well as datasets with non-uniform nonzero distributions that lead to load imbalance.
We also show that the performance benefits of the $s$-step methods are moderate as allreduce bandwidth becomes the dominant cost.
This observation underscores the importance of dataset characteristics and machine balance in determining the performance of the proposed methods in high-performance computing environments.
Our research extends prior work on $s$-step methods to kernelized machine learning models for classification and regression.
We show that in contrast to prior $s$-step coordinate descent and stochastic gradient descent methods, the kernel methods do not increase the total communication bandwidth (in theory) and attain speedups for a greater range of values of $s$.
In the future, we plan to further optimize the $s$-step methods' kernel computation and gradient correction overheads by approximating the sampled kernel matrix (for example using the Nystr\"{o}m method).
This performance optimization would enable the $s$-step method to scale to larger block sizes at the expense of weaker convergence.
We also aim to study the performance characteristics of the proposed methods in distributed environments (e.g. federated or cloud environments) where network latency costs are more prohibitive and where $s$-step methods may yield impactful performance improvements.

\section*{Acknowledgements}
The authors would like to thank Boyuan Pan for assistance with generating figures and Grey Ballard for helpful discussions and feedback on this manuscript.
This work was supported under Contract No. DE-AC05-00OR22725 with the US Department of Energy.
This research used resources of the National Energy Research Scientific Computing Center (NERSC), a U.S. Department of Energy Office of Science User Facility located at Lawrence Berkeley National Laboratory, operated under Contract No. DE-AC02-05CH11231 using NERSC award ASCR-ERCAP0024170.
This work also used computational resources provided by the Wake Forest University (WFU) High Performance Computing Facility for testing and debugging.

\bibliographystyle{siam}
\bibliography{refs}

\end{document}